%% file: main.tex
\documentclass[envcountsame]{llncs}

\usepackage[utf8]{inputenc}

\usepackage{xparse}
\usepackage{graphicx}
\usepackage{mathtools}
\usepackage{amsmath}
\usepackage{amssymb}
\usepackage{stmaryrd}
\usepackage{misc0}
\usepackage{stmaryrd}
\usepackage[capitalise,noabbrev]{cleveref}
\usepackage[ruled,vlined]{algorithm2e}
\usepackage{enumitem}

\usepackage{multicol}
\usepackage{thm-restate}
\usepackage[sort&compress, numbers]{natbib}
\newcommand{\alcformula}{\ensuremath{\varphi}\xspace}
\newcommand{\model}{\ensuremath{M}\xspace}

\renewcommand{\individuals}[1]{\ensuremath{{\sf ind}(#1)}\xspace}
\newcommand{\formulas}[1]{\ensuremath{{\sf f}(#1)}\xspace}
\newcommand{\concepts}[1]{\ensuremath{{\sf c}(#1)}\xspace}
\newcommand{\setqm}[1]{\ensuremath{{\sf S}(#1)}\xspace}
\renewcommand{\c}{\ensuremath{{\mathbf {c}}}\xspace}
\renewcommand{\f}{\ensuremath{{\mathbf {f}}}\xspace}
\newcommand\blfootnote[1]{%
  \begingroup
  \renewcommand\thefootnote{}\footnote{#1}%
  \addtocounter{footnote}{-1}%
  \endgroup
}

\title{Revising Ontologies via Models: \\ The \ALC-formula Case} 

\author{Jandson S.  Ribeiro\inst{1} \and
Ricardo Guimarães\inst{2} \and
Ana Ozaki\inst{2}}
\authorrunning{Ribeiro et al.}
\institute{Universität Koblenz-Landau, Germany \\
\email{jandson@uni-koblenz.de} \and
University of Bergen, Norway\\
\email{\{ricardo.guimaraes,ana.ozaki\}@uib.no}}

\begin{document}

\maketitle

\blfootnote{Copyright \textcopyright{} 2021 for this paper by its authors. Use permitted under
     Creative Commons License Attribution 4.0 International (CC BY 4.0).}

\begin{abstract}
  Most approaches for repairing description logic (DL) ontologies aim at
    changing the axioms as little as possible while solving inconsistencies,
    incoherences and other types of undesired behaviours. As in Belief
    Change, these issues are often specified using logical
    formulae. 
    Instead, in the new setting for updating DL ontologies that we propose here,
    the input for the change is given by a \emph{model} which we want to add
    or remove.
    The main goal is to minimise the loss of information, without concerning
    with the syntactic structure.
    This new setting is motivated by scenarios where an
    ontology is built automatically 
    and needs to be refined or updated.  In such situations, the syntactical form is often
    irrelevant and the incoming information is not necessarily given as
    a formula. We define general operations and conditions on which they are
    applicable, and instantiate our approach to the case of  \(\ALC\)-formulae.
\end{abstract}
\section{Introduction}

Formal specifications often have to be updated either due to modelling errors or
because they have become obsolete. When these specifications are
description logic (DL) ontologies, it is possible to use one of the many
approaches to fix missing or unwanted behaviours. Usually these methods involve
the removal or replacement of formulae responsible by the undesired
aspect~\cite{Suntisrivaraporn2009,Kalyanpur2006,Troquard2018,Baader2018,DBLP:conf/sum/OzakiP18}.

The problem of changing logical representations of knowledge upon the arrival of
new information is the subject matter of \emph{Belief Change}~\citep{Hansson1999}. The theory developed in
this field provides   constructions suitable for various formalisms and
applications~\cite{Hansson1999,Ribeiro2013,Santos2018,RibeiroNW19AAAI}.
In most approaches for Belief Change and for repairing ontologies, it is
assumed that a set of formulae represents the 
entailments to be added or removed. However, in some situations, it might be easier to obtain this
information as a \emph{model} instead. This idea relates  with Model
Checking~\citep{Clarke1986} whose main problem is to determine whether a model
satisfies a set of constraints; 
and with the paradigm of Learning
from Interpretations~\citep{DeRaedt1997}, where a formula needs to be created or 
changed so as to have certain interpretations as part of its models and 
remove others from its set of models.
\Cref{ex:modelInput} illustrates the intuition behind using
models as input.

\input{example.tex}

We propose a new setting for Belief Change, in particular,
  contraction and expansion functions which take models as input. We analyse
the case of
\(\ALC\)-formula using quasimodels as a mean to define belief change operations.
This logic satisfies properties which facilitate the design of these operations
and it is close to \(\ALC\), which is a well-studied DL.
Additionally, we identify
the postulates which determine these functions and prove that they characterise
the mathematical constructions via representation theorems.
The remaining of this work is organised as follows: in \cref{sec:beliefChange}
we introduce the concepts from Belief Change which our approach builds upon and
detail the paradigm we propose here. \Cref{sec:alc-case} presents
\(\ALC\)-formula, the belief operations that take models as input and their
respective representation theorems. In \cref{sec:relatedWorks}, we highlight
studies which share similarities with our proposal and we conclude in \cref{sec:conc}.

\section{Belief Change}
\label{sec:beliefChange}

\subsection{The Classical Setting}
Belief Change \citep{Alchourron1985,Hansson1999} studies the problem of how an
agent should modify its knowledge in 
 light of  new information.  
In the original paradigm of Belief Change,  the AGM theory,  
an agent's body of knowledge is represented as a set of formulae  closed under
logical consequence,  called a \emph{belief set}, and the new information is represented as
a single \emph{formula}. 
%
\textit{Belief sets}, however,  are not the only way for representing an agent's body of knowledge,  
and another way of representing an agent's knowledge is via \emph{belief bases}:  arbitrary sets of formulae,  not necessarily closed under logical consequence~\citep{Hansson1994a}.  
In the AGM paradigm,  
an agent may modify its current belief base $\Bmc$ in response to a new piece of information $\varphi$ through three  kinds of operations:
\begin{description}
	\item[Expansion:] $\Ex(\Bmc, \varphi)$, simply add $\varphi$ to $\baseb$;
    \item[Contraction:] $\Con(\Bmc, \varphi)$, reduce \(\baseb\) so that it
        does not imply \(\varphi\);
    \item[Revision:] $\Rev(\Bmc, \varphi)$, incorporate $\varphi$ and keep
        consistency of the resulting belief base, 
     as long as $\varphi$ 
     is consistent.  
\end{description}

When modifying its body of knowledge an agent should rationally modify its
beliefs conserving most of its original beliefs.  This principle of minimal
change is captured in Belief Change via sets of rationality postulates.  Each of
the three operations (expansion, contraction and revision) presents its own set of rationality postulates which characterize precisely different  classes of belief change constructions.  
The AGM paradigm was initially proposed for classical logics that satisfy
specific requirements, dubbed AGM assumptions, among them \textit{Tarskianicity,
compactness} and \emph{deduction}. See \citep{Flouris2006, Ribeiro2013} for a complete
list of the AGM assumptions and a discussion on the topic.    
Recently, efforts have been applied to extend Belief Change to logics that do
not satisfy such assumptions. For instance, logics that are not closed
under classical negation of formulae (such as is the case for most
DLs)~\citep{Ribeiro2013, Ribeirow2014}, and
temporal logics and logics without compactness~\citep{RibeiroNW18,
RibeiroNW19,RibeiroNW19AAAI}.


 In what follows,  we define \emph{kernel contraction}~\cite{Hansson1999}, one of the most studied
constructions in Belief Change and
which is closely related to the most common ways to repair ontologies.
Kernel operations rely on calculating the minimal implying sets (MinImps), also
known as justifications~\citep{Horridge2011} or kernels~\citep{Hansson1999}.
A MinImp is a minimal subset that does entail a formula $\varphi$.  
The set of all MinImps of a belief base $\baseb$ w.r.t.\ a formula $\varphi$ is denoted by
$\MinImps(\baseb, \varphi)$. A kernel contraction removes from each MinImp at least
one formula using an \emph{incision function}. 

\begin{definition}
    Given a set of formulae $\baseb$ of language \(\llang\), a function $f$ is an
    incision function for $\baseb$ iff for all $\varphi \in \llang$: (i)
    $f(\MinImps(\baseb, \varphi)) \subseteq \bigcup \MinImps(\baseb, \varphi)$ and (ii)
    $f(\MinImps(\baseb, \varphi)) \cap X \neq \emptyset$, for all $X \in
    \MinImps(\baseb, \varphi)$. 
\end{definition}

Kernel contraction operators are built upon incision functions.
\textcolor{black}{The application of an incision function to a set of MinImps is a hitting
set~\cite{Kalyanpur2006,Baader2018}.}

\begin{definition}
\label{def:kernelcon}
    Let \(\llang\) be a language and
    \(f\) an incision function. The \emph{kernel contraction} on \(\baseb \subseteq \llang\) determined by
    \(f\) is the operation \(\Con_f : 2^{\llang} \times \llang
    \mapsto 2^\llang \) defined as: \(\Con_f(\baseb, \varphi) = \baseb \setminus
    f(\MinImps(\baseb, \varphi))\).
\end{definition}

Kernel contraction operations are characterised precisely by a set of rationality postulates,   as shown in the following representation theorem: 

\begin{theorem}[\cite{Hansson2002}]
    Let \(\Cn\) be a consequence operator satisfying monotonicity and
    compactness defined for a language \(\llang\). Then \(\Con : 2^{\llang} \times \llang
    \mapsto 2^\llang\) is an operation of kernel contraction on \(\baseb \subseteq
    \llang\) iff for all sentences \(\varphi \in \llang\):
    \begin{description}
        \item[(success)] if \(\varphi \not\in \Cn(\emptyset)\), then \(\varphi \not\in
            \Cn(\Con(\baseb, \varphi))\),
        \item[(inclusion)] \(\Con(\baseb, \varphi) \subseteq \baseb\),
        \item[(core-retainment)] if \(\psi \in \baseb \setminus \Con(\baseb, \varphi)\),
            then there is some \(\baseb' \subseteq \baseb\) such that \(\varphi \not\in
            \Cn(\baseb')\) and \(\varphi \in \Cn(\baseb' \cup \psi)\),
        \item[(uniformity)] if for all subsets \(\baseb'\) of \(\baseb\), \(\varphi \in
            \Cn(\baseb')\) iff \(\psi \in \Cn(\baseb')\), then \(\Con(\baseb, \varphi)
            = \Con(\baseb, \psi)\).
    \end{description}
\end{theorem}

\subsection{Changing Finite Bases by Models}\label{sec:modelChange}
The Belief Change setting discussed in this section represents an epistemic
state by means of a finite base. While this essentially differ from the traditional
approach \citep{Alchourron1985,Hansson1994a}, it aligns with the KM paradigm
established by \citet{Katsuno1991}. In \cref{sec:relatedWorks} we discuss other
studies in Belief Change which also take finite representability into account.

In this work,  unlike the standard representation methods in Belief Change,   we consider that an incoming piece of information is represented as a finite model.  Belief Change operations defined in this format will be called model change operations.  
Recall that a model \(M\) is simply a structure used to give semantics to an underlying logic language.  
The set of all possible models is given by \(\mathfrak{M}\).   
Moreover, we assume a semantic system that, for each set of formulae
\(\baseb\) of the language \(\llang\) gives a set of models \(\modelsof{\baseb}
\coloneqq \{M \in \mathfrak{M} \mid \forall \varphi \in \baseb: M \models
\varphi\}\). 
\textcolor{black}{Let $\finitepwset(\llang)$ denote the set of all finite bases in $\llang$. }
 We also say that a set of models \(\mSet\) is \emph{finitely
representable in \(\llang\)} if there is a finite base \(\baseb \in
\finitepwset(\llang)\) such that \(\modelsof{\baseb} = \mSet\). 
Additionally, if for all  \(\varphi \in \llang\) it holds that \(M \models
\varphi\) iff \(M' \models \varphi\) then we write \(M \equiv^\llang M'\). We
also define \(\eqclass{M}{\llang} \coloneqq \{M' \in \mathfrak{M} \mid M' \equiv^\llang
M \}\). 

When compared to traditional methods in Belief Change and Ontology
Repair~\citep{Hansson1999,Kalyanpur2006,Baader2018}, where the incoming
information comes as a single formula, our approach receives instead a single
model as input. Although, the initial body of knowledge is represented as
a finite base, the operations we define do not aim to preserve its syntactic structure.



The first model change operation we introduce is model contraction,  which eliminates one of the models of the current base (which in Section~\ref{sec:alc-case} is instantiated as an ontology). 
Model contraction is akin to a belief expansion,  where a formula is added to the
belief set or base, reducing the set of accepted models.  The counterpart
operation, model expansion, changes the base to include a new model. This
relates to belief contraction, in which a formula is removed,
and thus more models are seen as plausible. 
 

We rewrite the rationality postulates that characterize kernel contraction~\citep{Hansson2002}, 
considering an incoming piece of information  represented as a model instead of a single formula.
%


\begin{definition}[Model Contraction]
\label{prepostcon}
    Let \(\llang\) be a language.
    A function \(\Con : \finitepwset(\llang) \times \mUni
    \mapsto \finitepwset(\llang)\) is a \emph{finitely representable model contraction function} iff for every \(\Bmc \in \finitepwset(\llang)\) and
    \(M \in \mUni\) it satisfies the following postulates:

    \begin{description} 
        \item [(success)] \(M \not \in \modelsof{\Con(\Bmc, M)}
        = \emptyset\), 
        \item [(inclusion)] \(\modelsof{\Con(\Bmc, M)} \subseteq
        \modelsof{\Bmc}\),
        \item [(retainment)] if \(M' \in \modelsof{\Bmc} \setminus
            \modelsof{\Con(\Bmc, M)}\) then \(M' \equiv^\llang M\),
        \item [(extensionality)] \(\Con(\baseb, M)
            = \Con(\baseb, M')\),   if \(M \equiv^\llang M'\).  
    \end{description}
\end{definition}


We might also need to add a model to the set of models of the
current \textcolor{black}{base}. 
This addition relates to classical contractions in Belief
Change, \textcolor{black}{which  \emph{reduces} the belief base}.

\begin{definition}[Model Expansion]
\label{prepostexp}
    Let \(\llang\) be a language. 
    A function \(\Ex : \finitepwset(\llang) \times \mUni
    \mapsto \finitepwset(\llang)\) is a \emph{finitely representable model
        expansion} iff for every \(\baseb \in \finitepwset(\llang)\) and
    \(M \in \mUni\) it satisfies the  postulates:
    \begin{description}
        \item[(success)] \(M \in \modelsof{\Ex(\baseb, M)},\)
        \item[(persistence)] \(\modelsof{\baseb} \subseteq \modelsof{\Ex(\baseb,
            M)},\)
        \item [(vacuity)] \(\modelsof{\Ex(\baseb,
            M)} = \modelsof{\baseb},\) if \(M \in \modelsof{\baseb}\),  
        \item [(extensionality)] \(\Ex(\baseb, M)
            = \Ex(\baseb, M')\),  if \(M \equiv^\llang M'\).  
    \end{description}
\end{definition}



\begin{definition}
\label{def:ideal}
    Let \(\llang\) be a language and \(\Cn\) a Tarskian consequence operator
    defined over \(\llang\). Also let \(\mUni\) be a fixed set of models. 
    We say that a triple \(\logsys = (\llang, \Cn, \mUni)\) is an \emph{\txtideal{} logical system} if
    the following holds.
    \begin{itemize}
        \item For every \(\baseb \subseteq \llang\) and \(\varphi \in \llang\), \(\baseb
            \models \varphi\) (i.e. \(\varphi \in \Cn(\baseb)\)) iff \(\modelsof{\baseb} \subseteq
            \modelsof{\varphi}\).
        \item For each   \(\mSet \subseteq
            \mUni\) there is a finite set of formulae \(\baseb\) such that
            \(\modelsof{\baseb} = \mSet\).
    \end{itemize}
\end{definition}

\textcolor{black}{%
If \(\logsys = (\llang, \Cn, \mUni)\) is an ideal logical system, we can define
a function \(\FR_\Lambda
: 2^\mUni \mapsto \finitepwset(\llang)\) and such
that \(\modelsof{\FR(\mSet)} = \mSet\). Then, we can define model
contraction as \(\Con(\baseb, M) = \FR(\modelsof{\baseb} \setminus
\eqclass{M}{\llang})\) and expansion as \(\Ex(\baseb, M)
= \FR(\modelsof{\baseb} \cup \eqclass{M}{\llang})\). 
The first condition in \Cref{def:ideal} implies that there
is a connection between the models satisfied and the logical consequences of the
base obtained and the second ensures that the result always exists.}
An example that fits
these requirements is to consider classical propositional logic with a finite
signature \(\Sigma\), together with its usual consequence operator and models.
\textcolor{black}{%
In this situation, we can define \(\FR_{prop}\) as follows:
\[
    \FR_{prop}(\mSet) = \bigvee_{M \in \mSet} \left( \bigwedge_{a \in \Sigma \mid
    M \models a} a \land \bigwedge_{a \in \Sigma \mid M \models \neg{a}} \neg{a} \right).
\]
}

Next, we show that the construction proposed with \(\FR\) has the properties
stated in \cref{prepostcon,prepostexp}.

\begin{theorem}
    Let \((\llang, \Cn, \mathfrak{M})\) be an \txtideal{} logical system as in \cref{def:ideal}. Then \(\iCon(\baseb, M) \coloneqq \FR\left(\modelsof{\baseb} \setminus
    \eqclass{M}{\llang}\right)\) satisfies the postulates in \cref{prepostcon}.
\end{theorem}

\begin{proof}
    \textcolor{black}{%
    \Cref{def:ideal} ensures that the result exists and that \(M \models
    \varphi\), for all \(\varphi \in \Lambda\), giving us success.
    By construction we do gain models, thus we have inclusion.}
    If \(M \equiv^\llang M'\),
    then \(\eqclass{M}{\llang}  = \eqclass{M'}{\llang}\), thus extensionality is
    satisfied. Also, if \(M' \in \modelsof{\alcformula}
    \setminus \modelsof{\iCon(\alcformula, M)}\) then \(M' \in
    \eqclass{M}{\llang}\), hence the operation satisfies retainment.
\end{proof}

\begin{theorem}
    Let \((\llang, \Cn, \mathfrak{M})\) be an \txtideal{} logical system as in \cref{def:ideal}. Then \(\iExp(\baseb, M) \coloneqq \FR\left(\modelsof{\baseb} \cup
        \eqclass{M}{\llang}\right)\) satisfies the
    postulates in \cref{prepostexp}.
\end{theorem}

\begin{proof}
    \textcolor{black}{%
    \Cref{def:ideal} ensures that the result exists and that \(M \models
    \varphi\), for all \(\varphi \in \Lambda\), giving us success.
    Due to the first condition in \Cref{def:ideal} we gain vacuity: if \(M \in
    \modelsof{\baseb}\), then there will be no changes in the accepted models.
    By construction we do not lose models, thus we have persistance.}
    Extensionality also holds because whenever \(M \equiv^\llang M'\) we have
    then \(\eqclass{M}{\llang}  = \eqclass{M'}{\llang}\).
\end{proof}

A revision operation incorporates new formulae,  and removes potential conflicts in behalf of consistency.  
In our setting,  incorporating information coincides with model contraction which could lead to an inconsistent belief state.  
In this case,  model revision could be interpreted as a conditional model
contraction: in some cases the removal might be rejected to
preserve consistency.
We leave the study on revision as a future work.  



\section{The case of \(\mathcal{ALC}\)-formula}\label{sec:alc-case}

The logic \(\mathcal{ALC}\)-formula corresponds to the DL \(\mathcal{ALC}\)
enriched with boolean operators over \ALC axioms. 
As discussed in \cref{sec:modelChange}, in finite representable logics,  such as the classical propositional logics, 
we can easily add and remove models while keeping the representation finite.  
   For $\ALC$-formula,  however,  it is not possible 
   to uniquely add or remove a new model $M$ since, for instance, the language does not distinguish
   quantities (e.g., a model $M$ and another model that has two duplicates of $M$). 

   Even if quantities are disregarded and our input is 
   a class of models indistinguishable by $\ALC$-formulae, 
   there are sets of formulae in this language that are not finitely
   representable. As for instance in the following infinite set: \(\{C \sqsubseteq \exists r^n.\top \mid
   n \in \mathbb{N}^{>0}\}\), where \(\exists r^{n + 1}.\top\) is a shorthand for
   \(\exists r.(\exists r^n.\top)\) and \(\exists r^1.\top \coloneqq \exists r.C\).
  As a workaround for the $\ALC$-formula case,  
   we propose a new strategy based on the translation of \(\mathcal{ALC}\)-formulae
into DNF.

\subsection{\ALC-formulae and Quasimodels}
Let \NC, \NR and \NI be countably infinite and pairwise disjoint sets of concept
names, role names, and individual names, respectively. \ALC 
\emph{concepts} are built according to the rule:
%
$C ::= A \mid \neg C \mid (C \sqcap C) \mid \exists r.C$, 
%
where $A \in \NC$ and \(r \in \NR\).
\ALC-\emph{formulae} are defined as expressions $\phi$ of the form 
\[\phi  ::= \alpha \mid \neg(\phi) \mid (\phi \wedge \phi) \quad\quad
\alpha ::= C(a) \mid r(a,b) \mid (C =\top),\]
%
where $C$ is an \ALC concept,  $a, b \in \NI$, and $r \in \NR$\footnote{
We may omit parentheses if there is no risk of confusion.
The usual  concept inclusions $C \sqsubseteq D$ can be expressed 
with  $\top \sqsubseteq \neg C\sqcup D$ and $\neg C\sqcup D \sqsubseteq \top$,
 which is $(\neg C\sqcup D = \top)$.}.  
%
Denote by $\individuals{\alcformula}$
the set of all individual names occurring in an \ALC-formula \alcformula.

The semantics of \ALC-formulae  and the definitions related to quasimodels are
standard~\cite[page 70]{Gabbay2003a}. \textcolor{black}{In what follows, we reproduce the
essential definitions and results for this work}.
Let $\alcformula$ be an \ALC-formula. 
Let  $\formulas{\alcformula}$ and $\concepts{\alcformula}$ be the set of all subformulae
and subconcepts of $\alcformula$ closed under single negation, respectively. 
 
\begin{definition}
A \emph{concept type} for $\alcformula$ is a subset 
$\c\subseteq \concepts{\alcformula}$ such that:
\begin{enumerate}
\item 
    $D \in \c$   iff $\neg D \not\in \c$, for all $D \in \concepts{\alcformula}$;
\item 
$D\sqcap E \in \c$   iff $\{D,E\}\subseteq \c$, for all $D\sqcap E \in \concepts{\alcformula}$. 
\end{enumerate}
\end{definition}
\begin{definition}
A \emph{formula type} for $\alcformula$ is a subset 
$\f\subseteq \formulas{\alcformula}$ such that:
\begin{enumerate}
\item 
    $\phi \in \f$ iff $\neg \phi \not\in \f$, for all $\phi \in\formulas{\alcformula}$;
\item 
$\phi\wedge \psi \in \f$ iff $\{\phi,\psi\}\subseteq \f$, for all $\phi\wedge \psi \in \formulas{\alcformula}$. 
\end{enumerate}
\end{definition}

We may omit `for $\alcformula$' if this is clear from the context.
A \emph{model candidate} for $\alcformula$ is 
 a triple $(T,o,\f)$ such that $T$ is a set of concept types, 
$o$ is a function from \individuals{\alcformula} to $T$,
\f a formula type, and 
$(T,o,\f)$
satisfies the conditions:  
$\alcformula\in\f$;
$C(a)\in\f$ implies $C\in o(a)$;
$r(a,b)\in\f$ implies $\{\neg C\mid \neg \exists r.C\in  o(a)\}\subseteq o(b)$.

\begin{definition}[Quasimodel]
A model candidate $(T,o,\f)$ for $\alcformula$ is a \emph{quasimodel} for $\alcformula$
if the following holds
\begin{itemize}
\item for every concept type   $\c \in T$ and every $\exists r.D \in \c$,   there is   $\c' \in T$ 
such that $\{D\}\cup\{\neg E\mid \neg \exists r.E\in \c\}\subseteq \c'$;
\item for every concept type   $\c \in T$ and every concept $C$, if $\neg C\in \c$
then this implies $(C=\top)\not\in\f$;
\item for every concept $C$, if $\neg(C =\top)\in \f$ then there is $\c \in T$
such that $C\not\in\c$;
\item $T$ is not empty.
\end{itemize}
\end{definition}

\Cref{qmSats} motivates the decision of using quasimodels to implement our
operations for finite bases described in \(\ALC\)-formulae.

\begin{theorem}[Theorem 2.27~\cite{Gabbay2003a}]
\label{qmSats}
An \(\mathcal{ALC}\)-formula \alcformula is satisfiable iff there is a quasimodel for \alcformula.
\end{theorem}

\subsection{\ALC-formulae in Disjunctive Normal Form}\label{subsec:alc-dnf}
\textcolor{black}{%
Next, we propose a translation method which converts an \ALC-formula into a disjunction of conjunctions of (possibly negated) atomic formulae.
}
Let $\setqm{\alcformula}$ be the set of all quasimodels for \alcformula. We
We define $\alcformula^\dagger$ as 

\[\bigvee_{(T,o, \f)\in \setqm{\alcformula}} (\bigwedge_{\alpha\in\f} \alpha\wedge \bigwedge_{\neg\alpha\in\f} \neg\alpha).\]
where $\alpha$ is of the form $(C=\top), C(a),r(a,b)$.




\textcolor{black}{\Cref{th:dagge_equals_phi} confirms the equivalence between
a formula and its translation into DNF. As downside, the translation can be
potentially exponentially larger than the original formula.}

\begin{restatable}{theorem}{alctranslation}\label{th:dagge_equals_phi}
For every \(\ALC\)-formula \alcformula, we have that $\alcformula\equiv\alcformula^\dagger$.
\end{restatable}

In the next subsections, we present finite base model change operations for
\ALC-formulae, i.e., functions from \(\llang \times \mUni \mapsto \llang\).
We can represent the body of knowledge as a single formula because every finite
belief base of \(\ALC\)-formulae can be represented by the conjunction of its
elements.
We use our translation to add models in a ``minimal'' way by \emph{adding disjuncts},
while removing a model amounts to \emph{removing disjuncts}. \textcolor{black}{We also
need to obtain a model candidate relative to our translated formula, as show in
\cref{def:qmOf}.}

\begin{definition}[\cite{Gabbay2003a}]
\label{def:qmOf}
    Let \(\I\) be an 
    interpretation and \(\alcformula\) an
    \(\ALC\)-formula formula. The quasimodel of \(\I\) w.r.t.\ \(\alcformula\),
symbols \(\qmOf{\alcformula}{\I}=(T,o,\f)\), is 
    \begin{itemize}
        \item \(T \coloneqq \{c(x) \mid x \in \Delta^\Imc\}\), where \(c(x) = \{C \in
            \concepts{\varphi} \mid x \in C^\I\}\), 
        \item \(o(a) \coloneqq c(a^\I)\), for all \(a \in \individuals{\varphi}\),
        \item \(\f \coloneqq \{\psi \in \formulas{\varphi} \mid \I \models \psi\}.\)
    \end{itemize}
\end{definition}

\subsection{Model Contraction for \(\mathcal{ALC}\)-formulae}

We define model contraction for \(\mathcal{ALC}\)-formulae using
the notion of quasimodels discussed previously and a correspondence between models and quasimodels.

We use the following operator, denoted $\qfilter$, to define model contraction in Definition~\ref{def:modelcontraction}. 
    Let \(\alcformula\) be an \(\ALC\)-formula and let $M$ be a model. 
    Then, 
\[ \qfilter(\varphi, M)  = \ftypes{\varphi} \setminus \{ \f\}, \mbox{ where } qm(\varphi,M) = (T,o,\f)\]
    and \(\ftypes{\alcformula}\) is
    the set of all formula types in all quasimodels for \(\alcformula\), that
    is:
    \[\ftypes{\alcformula} = \{\f \mid (T,o, \f)\in \setqm{\alcformula}\}.\]
Let $lit(\f) \coloneqq \{ \ell \in \f \mid \ell \mbox{ is a literal } \}$
be the set of all literals in a formula type $\f$.

\begin{definition}\label{def:modelcontraction}
A \emph{finite base model contraction function} is a function
 $\Con : \llang \times \mUni
    \mapsto \llang$ such that 
    \begin{align*}
\Con(\varphi, M) &{=} \left\{ 
\begin{array}{cl}
    \bigvee\limits_{\f \in \qfilter(\varphi,M)} \bigwedge lit(\f), & \mbox{if } M \models \varphi \mbox{ and }  \qfilter(\varphi, M) \neq \emptyset \\
    \bot & \mbox{if } M \models \varphi \mbox{ and } \qfilter(\varphi, M)  = \emptyset \\ 
     \varphi & \mbox{otherwise}. 
\end{array}
\right.
\end{align*}
\end{definition}





\input{contractionRepresentation.tex}

 

The postulate of \textit{success} guarantees that $M$ will be indeed relinquished,  
while \textit{inclusion} imposes that no model will be gained during a contraction operation.  
Recall that in order to guarantee finite representability,  it might be necessary to remove $M$ jointly with other models.  
The postulate \textit{atomic retainment}
captures a notion of minimal change,  dictating which models are allowed to be removed together with $M$.  

On the other hand,  \textit{atomic extensionality} imposes that if two models $M$ and $M'$ satisfy 
the same formulae within the literals of the current knowledge base $\varphi$,  then they should present the same result.  

A simpler way of implementing model contraction, also using the notion of
a quasimodel, 

\begin{definition}
    \label{fastcon}
    Let \(\alcformula\) be an \(\ALC\)-formula and \(M\) a model. Also, let \((T, o, \f)
    = \qmOf{\varphi}{M}\). The function  \(\Con_{s}(\alcformula, M)\) is defined follows:
    \begin{align*}
    \Con_{s}(\alcformula, M) = \left\{
        \begin{array}{cl}
            \alcformula \land \neg(\bigwedge lit(\f)) & \text{if } M \models \alcformula\\
        \alcformula & \text{otherwise.}
    \end{array}
    \right.
    \end{align*}
\end{definition}



\Cref{ex:contract} illustrates how $\Con_{s}$ works.  

{\color{black}
\begin{example}\label{ex:contract}
Consider the following $\ALC$-formula and interpretation $M$:
\begin{align*}
    \varphi \coloneqq & P(Mary) \land C(DL) \land C(AI) \land \left((teaches(Mary, DL)
        \right.
        \land \\
             &\left. \neg{teaches(Mary, AI)}) \lor (\neg{teaches(Mary, DL)} \land teaches(Mary, AI)) \right)
\end{align*}
and $M  = (\Delta^\Imc,\cdot^\Imc)$, where 
$\Delta^\Imc = \{  m,d,a \}$, 
$C^\Imc = \{ d,a\}$,
$P^\Imc = \{ m\}$,
$teaches^\Imc = \{(m,d)\}$,
$Mary^\Imc = m$, 
$AI^\Imc = a$, and
$DL^\Imc = d$.  
%
Assume we want to remove $M$ from $\modelsof{\varphi}$. 
Let $\qmOf{\varphi}{M} = (T,o, \f)$.  Thus,  
 \begin{align*}
  lit(\f) &= \{  \neg teaches(m,a),  teaches(m,  d), C(d), C(a), P(m) \}\\
\Con_{s}(\varphi,M) & = \varphi \land \neg  \bigwedge lit(\f) \\ 
	& = \varphi \land \neg \left( \neg teaches(m,a) \land  teaches(m,  d) \land C(d) \land C(a) \land  P(m)\right). 
 \end{align*}

\end{example}

}







Both model contraction operations $\Con$ and $\Con_{s}$ are equivalent.  

\begin{restatable}{theorem}{secondcon}
    For every  \(\ALC\)-formula \(\alcformula\) and  model $M$,
    \(\Con(\alcformula, M) {\equiv} \Con_{s}(\alcformula, M)\).
\end{restatable}


 
\subsection{Model Expansion in \(\mathcal{ALC}\)-formulae}

In this section, we investigate model expansion for $\ALC$-formulae. 
 Recall that we assume that a knowledge base is represented as a single $\ALC$-formula $\varphi$. 
  Expansion consists in adding an input model $M$ to the current knowledge base $\varphi$ with the 
  requirement that the new epistemic state can be represented also as a finite formula.  

\begin{definition}\label{def:rev:ALC}
Given a quasimodel $(T,o,\f)$, we write $\bigwedge (T,o,\f)$ as a short-cut  for $\bigwedge lit(\f)$.   
A \emph{finite base model expansion} is a function $\Ex: \llang \times \mUni \to \mathcal{L}$ s.t.: 
\begin{align*}
\Ex(\varphi, M) &= \left\{ 
\begin{array}{cl}
    \varphi & \mbox{if } M \models \varphi \\
    \varphi \lor \bigwedge qm(\neg \varphi, M) & \mbox{otherwise}. 
\end{array}
\right.
\end{align*}
\end{definition}

\Cref{ex:expand} illustrates how $\Ex$ works. 

\begin{example}\label{ex:expand}
Consider the interpretation $M$ from Example~\ref{ex:contract} and 
\[\varphi \coloneqq P(Mary) \land C(DL) \land C(AI) \land teaches(Mary,AI) \land
\neg{teaches(Mary, DL)}.\]
Assume we want to add $M$ to $\modelsof{\varphi}$ and  $qm(\neg \varphi,
M) = (T,o, \f)$. Thus,
 \begin{align*}
 lit(\f) &= \{  \neg teaches(m,a),  teaches(m,  d), C(d), C(a), P(m) \} \\
\Ex(\varphi,M) & = \varphi \lor \bigwedge lit(\f) \\ 
	& = \varphi \lor \left(\neg teaches(m,a) \land  teaches(m,  d) \land C(d)
        \land C(a) \land  P(m)\right). 
 \end{align*}

\end{example}

The operation `$\Ex$' maps a current knowledge base represented as a single formula $\varphi$ and maps it to a new knowledge base that is satisfied by the input model $M$.  
The intuition is that `$\Ex$' modifies the current knowledge base only if $M$ does not satisfy $\varphi$. 
 This modification is carried out by making a disjunct of $\varphi$ with a formula $\psi$ that is satisfied by $M$. 
 This  guarantees  that $M$ is present in the new epistemic state and that models of $\varphi$ 
 are not discarded.  
 The trick is to find such an appropriate formula $\psi$ which is obtained by taking the conjunction of 
 all the literals within the quasimodel  $qm(\neg \varphi, M)$. Here,  the quasimodel needs to be centred on $\neg \varphi$ because $M \not \models \varphi$,  and therefore it is not possible to construct a quasimodel based on $M$ centred on  $\varphi$.  
As discussed in the prelude of this section,  this strategy 
not only adds $M$ to the new knowledge base but also the whole equivalence class modulo the literals of  $\varphi$. 

\begin{restatable}{lemma}{expMmodelUnion}
\label{exp_model_union}
For every \ALC-formula $\varphi$ and model $M$: 
 \begin{align*} 
\modelsof{\Ex(\varphi, M)} = \modelsof{\varphi} \cup \liteqc{M}{\alcformula}.
 \end{align*}
\end{restatable}

Actually,  any operation that adds precisely the equivalence class of $M$ modulo the literals  is equivalent to `$\Ex$'.
\textcolor{black}{In the following, we write $\Ex^*(\varphi, M)$ to refer to an
arbitrary finite base expansion function of the form $\Ex^* : \llang \times
\mUni
    \mapsto \llang$.}

\begin{restatable}{theorem}{expModelConverse}
\label{expModelConverse}
For every   
 $\Ex^*$,  if $\modelsof{\Ex^*(\varphi, M)}
= \modelsof{\varphi} \cup \liteqc{M}{\alcformula}$ then  
\begin{itemize}
    \item[] (i) $\Ex^*(\varphi,M) \equiv \varphi$, if $M \models \varphi$; and
    \item[] (ii) $\Ex^*(\varphi,M) \equiv \varphi \lor \bigwedge qm(\neg \varphi, M)$, if $M \not \models \varphi$.
\end{itemize}
\end{restatable}

Our next step is to investigate the rationality of `$\Ex^*$'.  As expected
adding the whole equivalence class of 
\(M\) with respect to \(\langlit{\alcformula}\)
does not come freely,  and some rationality postulates are captured,  while others are lost:   

\begin{restatable}{theorem}{expansioncharacterization}
\label{expReprALCf}
Let \(M\) be a model and \(\alcformula\) an \(\ALC\)-formula. 
    A finite base model function \(\gExp(\alcformula, M)\) is equivalent to \(\oExp(\alcformula,
    M)\) iff \textcolor{black}{$\gExp$} satisfies:
    \begin{description}
        \item[(success)] \(\modelm \in \modelsof{\gExp(\alcformula, \modelm)}\).
        \item[(persistence):] \(\modelsof{\alcformula} \subseteq
            \modelsof{\gExp(\alcformula, \modelm)} \).
        \item[(atomic temperance):] For all \(\mSet' \subseteq \mUni\), if
            \(\modelsof{\alcformula} \cup \liteqc{\modelm}{\alcformula} \subseteq \mSet' \subset
            \modelsof{\gExp(\alcformula, \modelm)} \cup \{\modelm\}\) then
            \(\mSet'\) is not finitely representable in \(\ALC\)-formula.
        \item[(atomic extensionality)] {if \(M' \equivlit M\) then
            $$\modelsof{\gExp(\alcformula, M)} = \modelsof{\gExp(\alcformula,
        M')}.$$}
    \end{description}
\end{restatable}

The postulates \textit{success} and \textit{persistence} come from requiring that
 $M$ will be absorbed, and that models will not be lost during an expansion. 
  The \textit{atomic extensionality} postulate states that if two models satisfy exactly the same literals within $\varphi$, 
  then they should present the same results.  \textit{Atomic temperance} captures a 
  principle of minimality and guarantees that when adding $M$,  the loss 
  of information should be minimised. Precisely, the only formulae allowed to 
  be given up are those that are incompatible with $M$ modulo the literals of $\varphi$. 
  \Cref{exp_model_union} and \cref{expReprALCf} prove that the `$\Ex$' operation is characterized by the postulates:
  \textit{success,  persistence,  atomic temperance} and \textit{atomic
  extensionality}.

\input{relatedWorks.tex}
\section{Conclusion and Future Work}
\label{sec:conc}

In this work,  we have introduced a new kind of belief change operation: belief change via models.  
 In our approach,  an agent is confronted with a new piece of information in the format of a finite model, 
  and it is compelled to modify its current epistemic state, 
   represented as a single finite formula,  either incorporating the new model,  called model expansion; or 
   removing it,  called model contraction.    
The price for such finite representation is that the single input model cannot be removed or added alone, 
 and some other models must be added or removed as well. 
As future work, we will investigate model change operations in other DLs,
still taking into account finite representability. We will also explore
the effects of relaxing some constraints on Belief Base operations, allowing us to
rewrite axioms with different levels of preservation in the spirit of 
Pseudo-Contractions,  
 Gentle Repairs,  and 
 Axiom Weakening.

\bibliographystyle{plainnat}

\section*{Acknowledgements}

Part of this work has been done in the context of CEDAS (Center for Data Science, University of Bergen, Norway). 
The first author is supported by the German Research Association (DFG),  project number 424710479.
Guimarães is supported by the ERC project LOPRE (819416), led by Prof. Saket Saurabh.
Ozaki is supported by the Norwegian Research Council, grant number 316022.

\bibstyle{splncs04}
\bibliography{references.bib}

\appendix
\input{appendix}

\end{document}

%% file: example.tex
\begin{example}
    \label{ex:modelInput}
    Suppose that a system, which serves a university, uses
    an internal logical representation of the domain with a open world
    behaviour and unique names. Let \(\baseb\) be its current represention:
    \begin{align*}
        \baseb =& \left\{ Professors: \{Mary\}, Courses: \{DL, AI\}, \right.\\
                & \qquad \left. \{ teaches: \{(Mary, AI), (Mary, DL)\}\right\}.
    \end{align*}

    Assume that a 
    user finds mistakes in the course schedule and this is caused by the wrong
    information that Mary teaches the DL course. The user may lack knowledge to
    define the issue  formally. An alternative would be to provide the
    user with an interface where one can specify, for instance, that the
    following model should be accepted \(M = \{Professors = \{Mary\}, Courses
    = \{DL, AI\}, teaches = \{(Mary, AI)\}\}, \) (in this model Mary does not teach
    the DL course). Given this input, the system should repair itself
    (semi-)automatically.  
\end{example}

%% file: contractionRepresentation.tex
As we see later in this \lcnamecref{sec:alc-case}, there are models \(M,
M'\) such that \(M {\not\equiv}^\llang M'\) but our operations based on
quasimodels cannot distinguish them. 
Given \(\ALC\)-formulae \(\alcformula,\psi\),   
    we say that \(\psi\) is \emph{in the language of the literals of
    \(\alcformula\)},  written \(\psi \in \langlit{\alcformula}\), 
    if \(\psi\)
    is a boolean combination of the atoms in \(\alcformula\).
    Our operations partition the
models according to this restricted language.
    We  write \(M \equivlit M'\) instead \textcolor{black}{of \(M
    \equiv^\langlit{\alcformula} M'\)},  and \(\liteqc{M}{\alcformula}\) instead of
    \([M]^\langlit{\alcformula}\) for conciseness. 
 %

\begin{restatable}{theorem}{contractioncharacterization}
\label{conReprALCf}
Let \(M\) be a model and \(\alcformula\) an \(\ALC\)-formula. 
    A finite base model function \(\gCon(\alcformula, M)\) is equivalent to \(\oCon(\alcformula,
    M)\) iff \textcolor{black}{$\gCon$} satisfies:
    \begin{description}
        \item[(success)] $M \not \models \gCon(\alcformula, M)$,
        \item[(inclusion)] $\modelsof{\gCon(\varphi,M)} \subseteq \modelsof{\varphi}$,
        \item[(atomic retainment):] For all \(\mSet' \subseteq \mUni\), if \(\modelsof{\gCon(\baseb, \modelm)}
            \subset \mSet' \subseteq \modelsof{\baseb} \setminus \liteqc{\modelm}{\alcformula}\) then
            \(\mSet'\) is not finitely representable in \(\ALC\)-formula.
        \item[(atomic extensionality)] {if \(M' \equivlit M\) then
            $$\modelsof{\gCon(\alcformula, M)} = \modelsof{\gCon(\alcformula,
        M')}.$$}
    \end{description}
\end{restatable}

%% file: relatedWorks.tex
\section{Related Work}\label{sec:relatedWorks}

%
In the foundational paradigm of Belief Change, the AGM theory, 
bases have been used in the
literature with two main purposes: as a finite representation of the knowledge of an agent
\citep{Nebel1991,Dixon1993},  and as a way of distinguishing agents knowledge explicitly~\citep{Hansson1994a}.
Even though the AGM theory cannot be directly applied to   DLs
 because most of these logics do not satisfy the prerequisites known
as the AGM-assumptions~\citep{Flouris2005}, it 
has been  studied and adapted
to  DLs~\citep{Flouris2006,Ribeiro2008}. 
%

The syntactic connectivity 
in a knowledge 
base has a strong consequence of how an agent should modify its knowledge
 ~\citep{Hansson1999}. 
This sensitivity to syntax is also present in Ontology Repair and Evolution.  
Classical approaches preserve the syntactic form of the ontology as much as
possible~\citep{Kalyanpur2006,Suntisrivaraporn2009}. However, these approaches
may lead to drastic loss of information, as noticed by~\citet{Hansson1993}.
This problem has been studied in Belief Change for pseudo-contraction~\citep{Santos2018}.
%
%
In the same direction, \citet{Troquard2018} proposed the repair of DL ontologies by
weakening axioms using refinement operators.  
Building on this study, \citet{Baader2018} devised the theory of \emph{gentle repairs}, which also aims at
keeping most of the information within the ontology upon repair.  In fact,  
gentle repairs are closely related to pseudo-contractions~\citep{Matos2019}.


Other remarkable works in Belief Change in which the body of knowledge is
represented in a finite way include the formalisation of revision due to
\citet{Katsuno1991} and the base-generated operations by \citet{Hansson1996}. In
the former, \Citet{Katsuno1991} formalise traditional belief revision operations using a single
formula to represent the whole belief set. This is possible because they only
consider finitary propositional languages. \citeauthor{Hansson1996} provides
a characterisation of belief change operations over finite bases but restricted
for logics which satisfy all the AGM-assumptions (such as propositional classical logic).
%
%
\Citet{Guerra2019} develop operations for rational change where
an agent's knowledge or behaviour is given by a Kripke model. They also provide
two characterisations with AGM-style postulates.

%% file: appendix.tex
\section{Proofs for Section~\ref{sec:alc-case}}

We have already given the syntax of \ALC-formulae in the main text and
we provide the semantics here for the convenience of the reader. 
The semantics is given by interpretations. As usual, an interpretation \Imc is a pair 
$(\Delta^\Imc,\cdot^\Imc)$ where $\Delta^\Imc$ is a countable non-empty set, called the \emph{domain of invidivuals}, and $\cdot^\Imc$ is a
%
%
%
 function mapping each concept name $A\in\NC$ to a subset $A^\Imc$ of $\Delta^\Imc$
 and each role name $r\in\NR$ to a subset $r^\Imc$ of $\Delta^\Imc\times\Delta^\Imc$.
The interpretation of concepts  in \Imc is 
\begin{gather*}
\top^{\Imc} = \Delta^\Imc \quad (\neg C)^{\Imc} = \Delta^\Imc \setminus C^{\Imc} \quad 
(C \sqcap D)^{\Imc} = C^{\Imc} \cap D^{\Imc} \\
(\exists R.C)^{\Imc} = \{d \in \Delta^\Imc \mid \exists  d' \in C^{\Imc}: (d,d') \in R^{\Imc}\}.
\end{gather*}
The interpretation of formulae is as expected
\begin{gather*}
\Imc\models \neg (\varphi)\text{ iff not } \Imc\models \neg (\varphi) \quad\quad 
\Imc\models (\varphi \wedge \psi) \text{ iff } \Imc\models \varphi   \text{ and }\Imc\models  \psi \\
\Imc\models (C=\top)\text{ iff } C^\Imc=\top^\Imc 
\quad \Imc\models C(a) \text{ iff }a^\Imc\in C^\Imc
\quad \Imc\models r(a,b) \text{ iff } (a^\Imc,b^\Imc)\in r^\Imc. 
\end{gather*}

Formally, we inductively define the sets
$\formulas{\varphi}$ and $\concepts{\varphi}$   as follows.
\begin{definition}[Subformulae]
Given  an \ALC-formula $\varphi$, we have that 
\begin{itemize}
	\item if $\varphi$ is atomic then $\formulas{\varphi} \coloneqq \{ \varphi, \neg \varphi\}$;
	\item $\formulas{\varphi \land \psi} \coloneqq \{ \varphi \land \psi,  \neg(\varphi \land \psi)\} \cup \formulas{\varphi} \cup \formulas{\psi}$;
	\item $\formulas{\neg(\varphi \land \psi)} \coloneqq \formulas{\varphi \land \psi}$.
\end{itemize}
\end{definition}


\begin{definition}[Subconcepts]
Given a an \ALC-formula $\varphi$,  $\concepts{\varphi}$ is the minimal set satisfying the following conditions:
\begin{itemize}
	\item if $\textcolor{black}{(C =\top)} \in \formulas{\varphi}$ then $ C\in \concepts{\varphi}$;
	\item if $C(a) \in \formulas{\varphi}$ then $C \in \concepts{\varphi}$;
	\item if $C \sqcap D \in \concepts{\varphi}$ then $C, D \in \concepts{\varphi}$;
	\item if $\exists r. C \in \concepts{\varphi}$ then $C \in \concepts{\varphi}$;
	\item $\neg C \in \concepts{\varphi}$ iff $C \in \concepts{\varphi}$.
\end{itemize}
\end{definition}

\begin{definition}
    Let \((T, o, \f)\) be a model candidate for \(\alcformula\). Then, the
    interpretation \(\I_{(T, o, \f)}\) is defined as:
    \begin{itemize}
\item $\Delta^\Imc \coloneqq T\cup\individuals{\varphi}$;
\item $a^\Imc \coloneqq a$, for all $a\in\individuals{\varphi}$;
\item $A^\Imc \coloneqq \{\c \in T\mid A\in\c\}\cup \{a\in\individuals{\varphi}\mid A\in o(a) \}$;
\item $(\c,\c')\in r^\Imc$ iff $\{\neg C\mid \neg\exists r.C \in\c\}\subseteq \c'$, for $\c,\c'\in T$; 
\item $(a,b)\in r^\Imc$ iff $r(a,b)\in \f$, for $a,b\in \individuals{\varphi}$; 
\item $(a,\c)\in r^\Imc$ iff $\{\neg C\mid \neg\exists r.C \in o(a)\}\subseteq \c$, for $a\in\individuals{\varphi}$ and $\c\in T$.
\end{itemize}
\end{definition}

\begin{example}
    Let \(\varphi = (C=\top) \land \neg (C(a) \land \neg \textcolor{black}{(r(a, b))})\).
    Then, we have: 
    \begin{align*}
        \formulas{\varphi} =&\{\varphi, \neg\varphi, C = \top, C \neq \top, \neg (C(a) \land (\neg r(a, b))),\\
                    {}    & \qquad C(a) \land (\neg r(a, b)),  C(a), \textcolor{black}{\neg (C(a))}, \neg \textcolor{black}{(r(a, b))}, 
                    r(a, b)\}
    \end{align*}

    In any quasimodel \((T, o, \f)\) for \(\varphi\), we have that \(\varphi \in
    \f\). However this also implies that \((C = \top), \neg (C(a) \land (\neg r(a,
    b))) \in \f\). Consequently \((C(a) \land (\neg r(a,
    b))) \not\in \f\) and thus, \(C(a) \not\in \f\) or \(\neg r(a, b) \not\in \f\).
    Hence, there are only three possible formula types for \(\f\):
    
    \begin{math}
        \f =
        \begin{cases}
            \{\varphi, C = \top, \neg (C(a) \land (\neg r(a, b))), C(a), r(a,
            b)\} & \text{ or }\\
            \{\varphi, C = \top, \neg (C(a) \land (\neg r(a, b))), \neg \textcolor{black}{(C(a))}, r(a, b)\} & \text{ or }\\
            \{\varphi, C = \top, \neg (C(a) \land (\neg r(a, b))), \neg \textcolor{black}{(C(a))}, \neg \textcolor{black}{(r(a, b))}\}
        \end{cases}
    \end{math}

    Assuming that for each of these possible formula types there is at least one
    quasimodel of \(\f\), we get that:
    \begin{align*}
        \varphi^\dagger \equiv & \left((C = \top) \land C(a) \land r(a,b) \right) \textcolor{black}{\lor} \\
        {} & \left((C = \top) \land \neg \textcolor{black}{(C(a))} \land r(a,b) \right) \textcolor{black}{\lor} \\
        {} & \left((C = \top) \land \neg \textcolor{black}{(C(a))} \land \neg \textcolor{black}{(r(a, b))} \right)
    \end{align*}

    It is easy to check that the formula above is equivalent to \(\varphi\).

\end{example}

\begin{lemma}\label{obs:phi_sub_alpha}
\textcolor{black}{Let $\varphi,\phi$ be \ALC-formulae.}
If $\varphi \in \formulas{\phi}$ then $\formulas{\varphi} \subseteq \formulas{\phi}$. 
\end{lemma}
\begin{proof}
The proof follows by induction in the structure of $\phi$.  
	\noindent\textbf{Base:} $\phi$ is atomic.  
	Then,  by construction $\formulas{\phi} = \{ \phi, \neg \phi\} $. Thus, if $\varphi \in \formulas{\phi}$ then 
	$\varphi = \phi$ or $\varphi = \neg \phi$.  In either case,  $\formulas{\varphi} = \{ \varphi,  \neg \varphi\}$ 
	which implies that $\formulas{\varphi} = \{ \phi, \neg \phi\}$.  Thus, $\formulas{\varphi} \subseteq \formulas{\phi}$.  
	
	\textcolor{black}{In the following, assume  that $\phi$ is not atomic.  }
	

	\noindent\textbf{Induction Hypothesis:} 
	by construction $\phi$ is defined as the conjunction of two formulae $\psi$ and $\psi'$ or the negation of such conjunction,  that is, $\phi = \psi \land \psi'$ or $\varphi = \neg(\psi \land \psi')$.  
	Let us assume that for all $\beta \in \{ \psi,  \psi'\}$,  if $\varphi \in \formulas{\beta}$ then $\formulas{\varphi} \subseteq \formulas{\beta}$. 

		\noindent\textbf{Induction step:}  
	consider the cases (i) $\phi = \psi \land \psi'$ and (ii)  $\phi = \neg(\psi \land \psi')$.   
	\begin{itemize}
		\item[] (i)  $\phi = \psi \land \psi'$. 
		\textcolor{black}{By} construction 
		\begin{equation}
	\formulas{\phi} = \formulas{\psi \land \psi'} = \{ \psi \land \psi', \neg(\psi \land \psi')\} \cup \formulas{\psi} \cup \formulas{\psi'}. \label{eq: alpha_sub}	
		\end{equation}

	Thus,  (a) $\varphi \in \{ \psi \land \psi', \neg(\psi \land \psi')\}  $	 or (b) $\varphi \in \formulas{\psi}$ or (c) $\varphi \in \formulas{\psi'}$. 
	\begin{itemize}
		\item[] (a) $\varphi \in \{ \psi \land \psi', \neg(\psi \land \psi')\}$.  Thus, 
		 either $\varphi =  \psi \land \psi'$ or $ \varphi = \neg (\psi \land \psi')$.  For 
		 $\varphi =  \psi \land \psi'$, we get that $\formulas{\varphi} = \formulas{\psi \land \psi'} = \formulas{\phi}$ 
		 which means that $\formulas{\varphi} \subseteq \formulas{\phi}$. For $\varphi = \neg(\psi \land \psi')$, 
		 we get that $\formulas{\varphi} = \formulas{\neg(\psi \land \psi')}$. 
		 \textcolor{black}{By} construction,  $\formulas{\neg(\psi \land \psi')} = \formulas{\psi \land \psi'}$. 
		 Therefore,  $\formulas{\varphi} = \formulas{\psi \land \psi'} = \formulas{\phi}$ and so $\formulas{\varphi} \subseteq \formulas{\phi}$. 
		\item[] (b) $\varphi \in \formulas{\psi}$.   
		\textcolor{black}{By the inductive hypothesis,} 
		\textcolor{black}{$\formulas{\varphi} \subseteq \formulas{\psi}$}. 
		 From \eqref{eq: alpha_sub},   $\formulas{\psi} \subseteq \formulas{\phi}$. So, $\formulas{\varphi} \subseteq \formulas{\phi}$.
		\item[] (c) $\varphi \in \formulas{\psi'}$.  Analogous to item (b). 
	\end{itemize}		
		
%
		\item[] (ii) $\phi = \neg(\psi \land \psi')$.  
		\textcolor{black}{By} construction, $\formulas{\neg(\psi \land \psi')} = \formulas{\psi \land \psi'}$.
		So, $\formulas{\phi}  = 
	\formulas{\psi \land \psi'} \nonumber 
					 =  \{ \psi \land \psi', \neg(\psi \land \psi')\} \cup
                     \formulas{\psi} \cup \formulas{\psi'}$. 
		Proof proceeds as in item (i).  
	\end{itemize}
\end{proof}


\begin{lemma}\label{lem:interftype}
For every  $\ALC$-formula $\phi$  and formula type $\f$ \textcolor{black}{for} $\phi$,  if $\phi, \varphi \in \f$ 
then $\f \cap \formulas{\varphi}$ is a formula type \textcolor{black}{for} $\varphi$. 
\end{lemma}

\begin{proof}
\textcolor{black}{Let $\f_{\phi}$ be a fixed but arbitrary formula type for $\phi$ with  $\phi\in \f_{\phi}$.}
 We will show that $\f \coloneqq \f_{\phi} \cap \formulas{\varphi}$ is a formula type \textcolor{black}{for $\varphi$}.  
 Suppose for contradiction that $\f$ is not a formula type for $\varphi$.  Thus,  as $\f \subseteq \formulas{\varphi}$,  
 either condition (1) or (2) of the formula type definition is violated:
\begin{enumerate}
	\item There are formulae $\psi, \neg \psi \in \formulas{\varphi}$ such that either (a) $\psi \not \in \f$ and $\neg \psi \not \in \f$, or (b) $\psi, \neg \psi \in \f$. 
	\begin{enumerate}
		\item[] (a) $\psi \not \in \f$ and $\neg \psi \not \in \f$.
            \textcolor{black}{By} hypothesis,  $\varphi \in \f_{\phi}$. Thus, as
            $\f = \f_{\phi} \cap \formulas{\varphi}$, and by construction
            $\varphi \in \formulas{\varphi}$, we get that $\varphi \in \f$.
            Since $\f_{\phi}$ is a formula type,  we have that 
            for all \(\psi' \in \formulas{\phi}\),  \(\psi' \in \f_{\phi}\) iff
            \(\neg \psi' \not \in \f_{\phi}\). 
	As $\varphi \in \f_{\phi} \subseteq \formulas{\phi}$, it follows from Lemma~\ref{obs:phi_sub_alpha} that 
	$\formulas{\varphi} \subseteq \formulas{\phi}$. Therefore, for all \(\psi' \in \formulas{\varphi}\), \(\psi' \in \f_{\phi}\)
    iff \(\neg \psi' \not \in \f_{\phi}\). 
By hypothesis,  $\neg \psi, \psi \in \formulas{\varphi}$ which 
implies from above that either:
\begin{align}
\psi \in \f_{\phi} \mbox{ and } \neg \psi \not \in \f_{\phi},   \mbox{ or } \psi \not \in \f_{\phi} \mbox{ and } \neg \psi \in \f_{\phi}. \label{eq: falpha_con}
\end{align}

	By hypothesis,   $\neg \psi, \psi \in \formulas{\varphi}$ but $\neg \psi,
    \psi  \not \in \f$. Thus, as $\f = \f_{\phi} \cap \formulas{\varphi}$, we
    get $\neg \psi, \psi \not \in \f_{\phi}$, contradicting \eqref{eq: falpha_con}. 
		
		\item[] \textbf{(b) $\psi, \neg \psi \in \f$.  }By hypothesis, $\f_{\phi}$ is a formula type which implies that for all $\psi' \in \f_{\phi}$, $\psi', \neg \psi' \not \in \f_{\phi}$. Therefore, as $\f \subseteq \f_{\phi}$, we get that $\psi, \neg \psi \not \in \f$, a contradiction. 
		
\end{enumerate}	 
	\item Let $\psi \land \psi' \in \formulas{\varphi}$.  We will show that $\psi \land \psi' \in \f$ 
	iff $\{ \psi, \psi' \} \subseteq \f$ which contradicts the hypothesis that
    condition (2) from the formula type definition is violated.  We split the
    proof in two cases: either (a) $ \psi \land \psi' \in \f$ or (b) $\psi \land
    \psi' \not \in \f$.
		If $ \psi \land \psi' \in \f$,  as $\f = \f_{\phi} \cap
            \formulas{\varphi}$, we get that $\psi \land \psi' \in \f_{\phi}$.
            Since $\f_{\phi}$ is a formula type,  we have that $\{ \psi,
            \psi'\} \subseteq \f_{\phi}$.  By definition of
            $\formulas{\varphi}$, if $\psi \land \psi' \in \formulas{\varphi}$
            then $\{\psi, \psi'\} \subseteq \formulas{\varphi}$. Hence,  $\{ \psi, \psi'\} \in \f = \f_{\phi} \cap \formulas{\varphi}$.  

		Otherwise, $\psi \land \psi' \not \in \f$. As $\f = \f_{\phi} \cap \formulas{\varphi}$ and $\psi \land \psi' \in \formulas{\varphi}$, we get that $\psi \land \psi' \not \in \f_{\phi}$.  Thus, as $\f_{\phi}$ is a formula type,  we get that $\{ \psi, \psi'\} \not \subseteq \f_{\phi}$.  Therefore,  as $\f \subseteq \f_{\phi}$, we get that $\{ \psi, \psi'\} \not \subseteq \f$. 
	From (a) and (b) we conclude that $\psi \land \psi' \in \f$ iff $\{ \psi, \psi' \} \subseteq \f$. But 
	this contradicts the hypothesis that condition (2) from the formula type definition is violated. 
	
\end{enumerate}
Therefore, we conclude that $\f$ is a formula type. 

\end{proof}

\begin{lemma}\label{lem:neg}
For every \textcolor{black}{$\ALC$-formula $\varphi$, 
$\formulas{\varphi} = \formulas{\neg \varphi}$}
\footnote{
\textcolor{black}{
We silently remove double negation and treat $\neg\neg\phi$ as equal to $\phi$.}}.
\end{lemma}

\begin{proof}
\textcolor{black}{By construction $\varphi$ is a subformula of $\neg \varphi$. We can see that
$\formulas{\neg \varphi}=\formulas{\varphi} \cup \{\neg \varphi\}$.
Since $\formulas{\varphi}$ is closed under single negation
and, by construction, $\varphi\in \formulas{\varphi}$,
we have that $\neg \varphi\in \formulas{\varphi}$.
Thus, $\formulas{\varphi} = \formulas{\neg \varphi}$.}
\end{proof}

\begin{definition}
Let $\varphi$ be an $\ALC$-formula. The set of of formula types for $\varphi$ that has $\varphi$ is given by the set 
$$\tau(\varphi) = \{ \f \subseteq \formulas{\varphi} \mid \f \mbox{ is a formula type for } \varphi \mbox{ and } \varphi \in \f \}.$$
\end{definition}

\begin{lemma}\label{lem:observation}
For every $\ALC$-formula $\phi$  and formula type $\f$ for $\phi$,  if $\phi \in \f$ and 
 $\varphi \in \formulas{\phi}$ then $\f \cap \formulas{\varphi} \in \tau(\varphi) \cup \tau(\neg \varphi)$. 
\end{lemma}

\begin{proof}
Let $\f_{\phi}$ be a fixed but arbitrary formula type  for $\phi$ with $\phi \in \f_{\phi}$. 
As $\f_{\phi}$ is a formula type (for $\phi$) and $\varphi \in \formulas{\phi}$, 
 either (i) $\varphi \in \f_{\phi}$ or $\neg \varphi \in \f_{\phi}$: 
\begin{enumerate}
	\item[] (i) $\varphi \in \f_{\phi}$.  Thus,  by Lemma~\ref{lem:interftype}, we have that 
	$\f_{\phi} \cap \formulas{\varphi}$ is a formula type of $\varphi$.  Also, $\varphi \in \f_{\phi} \cap \formulas{\varphi}$. 
	Therefore, $\f_{\phi} \cap \formulas{\varphi} \in \tau(\varphi)$ which means that $\f_{\phi} \cap \formulas{\varphi} \in \tau(\varphi) \cup \tau(\neg \varphi)$. 
	\item[] (ii) $\neg \varphi \in \f_{\phi}$.  Thus,  by Lemma~\ref{lem:interftype}, we have that 
	$\f_{\phi} \cap \formulas{\neg \varphi}$ is a formula type for $\neg \varphi$.  
	Also, $\neg \varphi \in \f_{\phi} \cap \formulas{\neg \varphi}$. Therefore, 
	$\f_{\phi} \cap \formulas{\neg \varphi} \in \tau(\neg \varphi)$ which means that 
	$\f_{\phi} \cap \formulas{\neg \varphi} \in \tau(\varphi) \cup \tau(\neg \varphi)$. 
	By Lemma~\ref{lem:neg}, we have that $\formulas{\varphi} = \formulas{\neg
    \varphi}$ which implies that $\f_{\phi} \cap \formulas{\neg \varphi}
    = \f_{\phi} \cap \formulas{\varphi}$.  Therefore,  $\f_{\phi} \cap
    \formulas{\varphi} \in \tau(\varphi) \cup \tau(\neg \varphi)$. 
\end{enumerate}
\end{proof}

\begin{lemma}\label{lem:subtypes}
For every $\ALC$-formula $\varphi$,   $\f \in \tau(\varphi)$ iff $\f$ is a formula type for $\varphi$ and
\begin{enumerate}
	\item if $\varphi$ is atomic then $\f = \{ \varphi\}$;
	\item if $\varphi = \psi \land \psi'$ then $\f = \{  \psi \land \psi'\} \cup \f_{\psi} \cup \f_{\psi'} $,  for some $\f_{\psi} \in \tau(\psi)$ and $\f_{\psi'} \in \tau(\psi')$;
	\item if $\varphi = \neg (\psi \land \psi')$ then $\f = \{  \neg(\psi \land \psi')\} \cup \f_{\psi} \cup \f_{\psi'} $,  for some 
	$\f_{\psi} \in \tau(\psi) \cup \tau(\neg \psi),  \f_{\psi'} \in \tau(\psi') \cup \tau(\neg \psi')$ such that either $\f_{\psi} \in \tau(\neg \psi)$ or $\f_{\psi'} \in \tau(\neg \psi')$.   
\end{enumerate}
\end{lemma}
\begin{proof}
The direction ``$\Leftarrow$'' is trivial, so we focus only on the ``$\Rightarrow$''  direction.  Let $\f \in \tau(\varphi)$.  
Thus,  $\varphi \in \f$ and $\f$ is a formula type for $\varphi$. By construction,  (I) either $\varphi$ is atomic or (II) $\varphi = \psi \land \psi'$ or (III) $\varphi = \neg (\psi \land \psi')$:
\begin{enumerate}
	\item[] (I) $\varphi$ is atomic.  Thus, by construction $\f = \{ \varphi\}$ or $\f = \{ \neg \varphi\}$.  
	Thus,  as $\varphi \in \f$, we get  $\f = \{  \varphi\}$.  
	\item[] (II) $\varphi = \psi \land \psi'$.   As $\varphi \in \f$,  
	we get that $\psi \land \psi' \in \f$.  Moreover,  as $\f$ is a formula type for
	 $\varphi$ and $\psi \land \psi' \in \f$, it follows that $\psi, \psi' \in \f$.  

     Let 
	\begin{align*}
	\f_{\psi} \coloneqq \f \cap \formulas{\psi} & \mbox{ and } \f_{\psi'} \coloneqq \f \cap \formulas{\psi'}. 
	\end{align*}
	
	 As $\psi, \psi' \in \f $ and $\f$ is a formula type for $\varphi = \psi \land \psi'$,  by 
	 Lemma~\ref{lem:interftype}, $ \f_{\psi} = \f \cap \formulas{\psi} $ is a formula type for $\psi$ 
	 and $ \f_{\psi'} = \f \cap \formulas{\psi'} $ is a formula type for $\psi'$.  We have that 
	 $\psi \in \formulas{\psi}$ and $\psi' \in \formulas{\psi'}$ which means that  
	 $\psi \in \f_{\psi}$ and $\psi' \in \f_{\psi'}$.  Thus,  $\f_{\psi} \in \tau(\psi)$ and $\f_{\psi'} \in \tau(\psi')$.  
	  We still need to show that $\f = \{  \psi \land \psi'\} \cup \f_{\psi} \cup \f_{\psi'} $. For this, we will show that
		(i) $\f \subseteq \{  \psi \land \psi'\} \cup \f_{\psi} \cup \f_{\psi'} $ and (ii) 
		$\{  \psi \land \psi'\} \cup \f_{\psi} \cup \f_{\psi'} \subseteq \f$. The case (ii) is trivial, 
		so we focus only on case (i).  Let $\phi \in \f$.  As $\f$ is a formula type for $\varphi = \psi \land \psi'$, we get that 
		$$ \phi \in \f \subseteq \formulas{\psi \land \psi'} = \{ \psi \land \psi', \neg (\psi \land \psi')\} 
		\cup \formulas{\psi} \cup \formulas{\psi'}. $$ 
		%
		Therefore,   (a) $\phi \in \{ \psi \land \psi', \neg (\psi \land \psi')\}$ or (b)
		 $\phi \in \formulas{\psi}$ or (c) $\phi \in \formulas{\psi'}$. 
		\begin{itemize}
			\item[] (a)  $\phi \in \{ \psi \land \psi', \neg (\psi \land
                \psi')\}$.  As $\f$ is a  formula type and $\varphi = \psi \land
                \psi' \in \f$, we get that $\neg (\psi \land \psi') \not \in
                \f$.  Therefore,  as $\phi \in \f$, we have that $\phi \neq \neg(\psi
                \land \psi')$. Hence, $\phi = \psi \land \psi'$,  which implies that 
			$\phi \in \{ \psi \land \psi'\} \cup \f_{\psi} \cup \f_{\psi'}$.  
			\item[] (b)  $\phi \in \formulas{\psi}$.  Thus,  as $\phi \in \f$, we get that 
			$\phi \in \f_{\psi} = \f \cap \formulas{\psi}$ which implies  that $\phi \in \{ \psi \land \psi'\} 
			\cup \f_{\psi} \cup \f_{\psi'}$.
			\item[] (c) $\phi \in \formulas{\psi'}$.  Thus,  as $\phi \in \f$, we get that $\phi \in \f_{\psi'} = \f \cap \formulas{\psi'}$ which implies  that $\phi \in \{ \psi \land \psi'\} \cup \f_{\psi} \cup \f_{\psi'}$.  
		\end{itemize}
		Thus,  $\phi \in \{ \psi \land \psi'\} \cup \f_{\psi} \cup \f_{\psi'}$.  
	
	\item[] (III) $\varphi = \neg (\psi \land \psi')$.   As $\varphi \in \f$,  we get that $\neg(\psi \land \psi') \in \f$.  
	Let 
	\begin{align*}
	\f_{\psi} \coloneqq \f \cap \formulas{\neg \psi} & \mbox{ and } \f_{\psi'} \coloneqq \f \cap \formulas{\neg \psi'}. 
	\end{align*}
	
	As $\neg \psi, \neg \psi' \in \formulas{\varphi = \neg(\psi \land \psi')}$,  by Lemma~\ref{lem:observation}, we have that 
	$$ \f_{\psi} \in \tau(\psi) \cup \tau(\neg \psi) \mbox{ and } \f_{\psi'} \in \tau(\psi') \cup \tau(\neg \psi'). $$

Moreover,  as $\f$ is a formula type for $\varphi$ and $\varphi = \neg(\psi \land \psi') \in \f$, it 
follows that $\psi\land \psi' \not \in \f$. Therefore, $\{ \psi, \psi'\} \not \subseteq \f$.   
Thus, either $\psi \not \in \f$ or $\psi' \not \in \f$.  Thus,  as $\f$ is a formula type,  
either (i) $\neg \psi \in \f$ or (ii) $\neg \psi' \in \f$.  	
\begin{itemize}
	\item[] (i) $\neg \psi \in \f$.   Thus,  as $\f$ is a formula type for $\varphi = \neg (\psi \land \psi')$,  
	by Lemma~\ref{lem:interftype}, $ \f_{\psi} = \f \cap \formulas{\neg \psi} $ is a formula type for $\neg \psi$.  
	We have that $\neg \psi \in \formulas{\neg \psi}$. So  $\neg \psi \in \f_{\psi}$.  Thus,  
	$\f_{\psi} \in \tau(\neg \psi)$.   
	
	\item[] (ii) $\neg \psi' \in \f$.  	Analogously to item (i), we get that  $\f_{\psi'} \in \tau(\neg \psi')$.   
\end{itemize}
Thus,
$$ \f_{\psi} \in \tau(\neg \psi) \mbox{ or } \f_{\psi'} \in \tau(\neg \psi').$$

We still need to show that $\f = \{  \neg(\psi \land \psi')\} \cup \f_{\psi} \cup \f_{\psi'} $.  For this we need to show that (i)  $\f \subseteq \{  \neg(\psi \land \psi')\} \cup \f_{\psi} \cup \f_{\psi'} $ and (ii) $\{  \neg(\psi \land \psi')\} \cup \f_{\psi} \cup \f_{\psi'} \subseteq \f$.  The case (ii) is trivial.  So we focus only on case (i).  
	
Let $\phi \in \f$.  As $\f$ is a formula type for $\varphi = \neg(\psi \land \psi')$, we get that 
\begin{align*}
\phi \in \f \subseteq \formulas{\neg(\psi \land \psi')} & = \formulas{\psi \land \psi'}\\
	&= \{   \psi \land \psi',  \neg(\psi \land \psi') \} \cup \formulas{\psi} \cup \formulas{\psi'}. 
\end{align*}
		Therefore,   (a) $\phi \in   \{   \psi \land \psi',  \neg(\psi \land \psi') \}$ or (b) $\phi \in \formulas{\psi}$ or (c) $\phi \in \formulas{\psi'}$. 	
	\begin{itemize}
			\item[] (a)  $\phi \in \{ \psi \land \psi', \neg (\psi \land \psi')\}$.  As $\f$ is a  formula type and $\varphi = \neg(\psi \land \psi') \in \f$, we get that $ (\psi \land \psi') \not \in \f$.  Therefore,  as $\phi \in \f$, we have that $\phi \neq (\psi \land \psi')$. Therefore, $\phi = \neg (\psi \land \psi') $, which implies that 
			$\phi \in \{ \neg(\psi \land \psi')\} \cup \f_{\psi} \cup \f_{\psi'}$.  
			\item[] (b)  $\phi \in \formulas{\psi}$.  By Lemma~\ref{lem:neg}, we get $\formulas{\psi} = \formulas{\neg \psi}$. Therefore, $\phi \in \formulas{\neg \psi}$. 
			Thus,  as $\phi \in \f$, we get that $\phi \in \f_{\psi} = \f \cap \formulas{\neg \psi}$ which implies  that $\phi \in \{ \neg(\psi \land \psi')\} \cup \f_{\psi} \cup \f_{\psi'}$
			\item[] (c) $\phi \in \formulas{\psi'}$.  Analogously to item (b), we get  
			$\phi \in \{ \neg(\psi \land \psi')\} \cup \f_{\psi} \cup \f_{\psi'}$.  
		\end{itemize}
		Thus,  $\phi \in \{ \neg(\psi \land \psi')\} \cup \f_{\psi} \cup
        \f_{\psi'}$. 
\end{enumerate} 
\end{proof}



\begin{definition}[Formula degree]\textcolor{black}{
The degree of an $\ALC$-formula $\phi$, denoted $degree(\phi)$, is 
\begin{itemize}
	\item $1$ if $\phi$ is an atomic $\ALC$-formula; 
	\item $degree(\varphi) +1$ if $\phi=\neg \varphi$; and
	\item $degree(\varphi) + degree(\psi)$ if $\phi=\varphi \land \psi$.
\end{itemize}}
\end{definition}


\begin{lemma}[\cite{Gabbay2003a}]
\label{IphiQMphi}
    If \(\I \models \alcformula\) then \(\qmOf{\alcformula}{\I}\) is
    a quasimodel for \(\alcformula\).
\end{lemma}


To show Theorem~\ref{th:dagge_equals_phi}, we use Lemma~\ref{litToForm}. 

\begin{restatable}{lemma}{litToFormlemma}
\label{litToForm}
Let $\varphi$ be an $\ALC$-formula. If $\f \in \tau(\varphi)$  then
$$\bigg( \bigwedge lit(\f)  \bigg) \models \varphi.$$
\end{restatable}
\begin{proof}
 The proof follows by induction in the degree of $\phi$. \\
 
 \textbf{Base:} $degree(\phi) = 1$.  Then $\phi$ is atomic. 
  This implies from Lemma~\ref{lem:subtypes} that $\f = \{ \phi\}$.  Thus,  $\bigwedge lit(\f) = \phi$. 
  For \ALC, this means that $\bigwedge lit(\f) \models \phi$.  \\
 
 \textbf{Induction Hypothesis:} For every formula $\varphi$,  and formula type $\f_{\varphi}$  for $\varphi$,  if $\varphi \in \f_{\varphi}$ and $degree(\varphi) < degree(\phi)$ then $\bigwedge lit(\f_\varphi) \models \varphi$. \\ 
 
 \textbf{Induction Step:} Let $degree(\phi) > 1$.  By construction, $\phi$ is of the form $\varphi \land \psi$ or $\neg \varphi$, for some $\ALC$-formulae $\varphi$ and $\psi$: 
 
 \begin{enumerate}
 	\item $\phi = \varphi \land \psi$.  Thus,  from Lemma~\ref{lem:subtypes},  
 	$$ \f = \{\varphi \land \psi\}\cup \f_{\varphi} \cup \f_{\psi}, \mbox{ such that } 
 	\f_{\varphi} \in \tau(\varphi), \f_{\psi} \in \tau(\psi).$$  	 	
	Note that $lit(\f) = lit(\f_{\varphi}) \cup lit(\f_{\psi})$. Therefore, 
	$$ \bigwedge lit(\f) = \bigg( \bigwedge lit(\f_{\varphi}) \bigg) \land \bigg( \bigwedge lit(\f_{\psi}) \bigg)$$

By the definition of degree, we get that 
$degree(\phi) = degree(\varphi \land \psi) = degree(\varphi) + degree(\psi) $ and $1 \leq degree(\varphi) $ and $1 \leq degree(\psi)$.  
Therefore,  $degree(\varphi) < degree(\phi)$ and $degree(\psi) < degree(\phi)$.  By the inductive hypothesis,
$$  \bigwedge lit(\f_{\varphi}) \models \varphi \mbox{ and }  \bigwedge lit(\f_{\psi}) \models \psi. $$
Therefore, 
$$ \bigwedge lit(\f) =   \bigwedge lit(\f_{\varphi}) \land  \bigwedge lit(\f_{\psi}) \models \varphi \land \psi.$$
Thus,  as $\phi = \varphi \land \psi$ , we get 
$$ \bigwedge lit(\f) \models \phi.$$

	\item $\phi = \neg \varphi$.  By construction, either: (a) $\varphi$ is atomic,  or (b)$\varphi = \psi \land \psi'$.  
	\begin{enumerate}
		\item[(a)] 	$\varphi$ is atomic.  We get from Lemma~\ref{lem:subtypes} that $\f = \{ \neg \varphi \}$, which implies that $lit(\f) = \{ \neg \varphi \}$, and analogous to the base case, we get that $\bigwedge lit(\f) \models \neg \varphi$ that is, $\bigwedge lit(\f) \models \phi$. 
		\item[(b)] $\varphi = \psi \land \psi'$.   By Lemma~\ref{lem:subtypes}, we get that 
\begin{align}
\f = \{  \neg(\psi \land \psi')\} \cup \f_{\psi} \cup \f_{\psi'},  \label{eq:lem_wedge_case2}
\end{align}		
where 
	$\f_{\psi} \in \tau(\psi) \cup \tau(\neg \psi),  \f_{\psi'} \in \tau(\psi') \cup \tau(\neg \psi')$ such that either 
	
	$$\f_{\psi} \in \tau(\neg \psi) \mbox{ or }\f_{\psi'} \in \tau(\neg \psi'). $$	 
	
\begin{enumerate}
	\item  $\f_{\psi} \in \tau(\neg \psi)$.  	
	From definition of degree,  we get that 	
$$degree(\phi) = degree(\neg (\psi \land \psi')) = degree(\psi) + degree(\psi') + 1,  $$
and 	$degree(\psi) \geq 1$ and $degree(\psi') \geq 1 $ and  $degree(\neg \psi) = degree(\psi) +1$. Thus, 
$ degree(\phi) = degree(\neg (\psi \land \psi')) = degree(\neg \psi) + degree(\psi')$. Thus, as $degree(\psi') \geq 1$ we get 
 $$ degree(\neg \psi) < degree(\phi).  $$ 
 Thus,  by the inductive hypothesis,  $\bigwedge lit(\f_{\psi}) \models \neg \psi $.  Note that for every formula $\beta$,  $\neg \psi \models \neg(\psi \land \beta)$. Therefore,  for $\beta = \psi'$
 $$\bigwedge lit(\f_{\psi}) \models \neg (\psi \land \psi')$$

From \eqref{eq:lem_wedge_case2}, we get that 
$$ \bigwedge lit(\f) = \bigwedge lit(\f_{\psi}) \land \bigwedge lit(\f_{\psi'}).  $$
Thus,  as $\bigwedge lit(\f_{\psi}) \models \neg (\psi \land \psi')$, we get  that $ \bigwedge lit(\f_{\psi}) \land \bigwedge lit(\f_{\psi'}) \models \neg(\psi \land \psi')$ which implies from above that 
$ \bigwedge lit(\f) \models \neg(\psi \land \psi') $
that is, 
$$ \bigwedge lit(\f) \models \phi. $$
\item   $\f_{\psi} \in \tau(\neg \psi')$.  Analogous to item (i).   	
\end{enumerate}	
	\end{enumerate}
 \end{enumerate}
\end{proof}

\secondcon*

\begin{proof}
    If \(M \not\models \alcformula\), then \(\Con(\alcformula, M)
    = \Con'(\alcformula, M) = \alcformula\). Now, suppose that \(M \models
    \alcformula\). In this case, we know from \cref{IphiQMphi} that \((T', o',
    \f') = \qmOf{\alcformula}{M}\) is a quasimodel for \(\alcformula\).
    From \cref{th:dagge_equals_phi}, we know that:
    \begin{align*}
        \Con'(\alcformula, M) &\equiv \alcformula^\dagger \land \neg(\bigwedge lit(\f))\\
        {} &= \left(\bigvee_{(T,o, \f)\in \setqm{\alcformula}} (\bigwedge lit(\f))\right) \land \neg(\bigwedge lit(\f))
    \end{align*}

    For each \((T, o, \f)\) in \setqm{\alcformula} either \(lit(\f)
    = lit(\f')\) or \(lit(\f) \neq lit(\f')\). If \(lit(\f) = lit(\f')\), then:
    \[\bigwedge lit(\f) \land \neg(\bigwedge lit(\f')) \equiv \bot.\]
    Otherwise, we know that 
    \[\bigwedge lit(\f) \land \neg(\bigwedge lit(\f')) \not\equiv \bot.\]

    Due to the definition of $\qfilter$ and \cref{eqLitEqf} we can conclude that 
    for every \(\f \in \ftypes{\alcformula}\), \(\f \in \qfilter(\varphi, M)\) iff
    \(lit(\f) \neq lit(\f')\). As we can ignore inconsistent formulae in
    disjunctions, we get that \(\Con(\alcformula, M) \equiv \Con'(\alcformula,
    M)\).
\end{proof}


\alctranslation*
\begin{proof}
    Let \(\alcformula\) be an \(\ALC\)-formula and \(\I\) an
    interpretation.

    \paragraph{\(\I \models \alcformula \Rightarrow \I \models \alcformula^\dagger\):}
    First, suppose that \(\I \models \alcformula\). From \cref{IphiQMphi} we know
    that \(\qmOf{\alcformula}{\I} = (T, o, \f)\) is a quasimodel of \(\alcformula\). Therefore, there
    is a disjunct \(\psi\) of \(\alcformula^\dagger\) which is the conjunction of
    all atomic formulae in \(\f\). By \cref{def:qmOf} \(\I \models \f\), thus we
    can conclude that 
    \(\I \models \alcformula^\dagger\).
    
    \paragraph{\(\I \models \alcformula^\dagger \Rightarrow \I \models \alcformula\):}
    Now, assume that \(\I \models \alcformula^\dagger\). This means that there is one
    disjunct \(\psi\) of \(\alcformula^\dagger\) such that \(\I \models \psi\). By
    construction, this disjunct is a conjunction of atomic formulae in the formula type
    of a quasimodel \((T, o, \f)\) for \(\alcformula\). Using \cref{litToForm} we
    can conclude that \(\I \models \f\). As \(\alcformula \in \f\) we get that
    \(\I \models \alcformula\). Hence, \(\I \models \alcformula\) iff
    \(\I \models \alcformula^\dagger\), i.e., \(\alcformula \equiv
    \alcformula^\dagger\).
\end{proof}

\Cref{eqLitEqf} is a direct consequence of the definition of 
a formula type.

\begin{corollary}
\label{eqLitEqf}
    Let \((T, o, \f)\) and \((T', o', \f')\) be  quasimodels for an
    \(\ALC\)-formula   \(\alcformula\). Then,  \(lit(\f) = lit(\f')\) iff
    \(\f = \f'\).
\end{corollary}

    Given    \(\ALC\)-formulae \(\alcformula,\psi\),   
    we say that \(\psi\) is in the language of the literals of
    \(\alcformula\),  written \(\psi \in \langlit{\alcformula}\), 
    if \(\psi\)
    is a boolean combination of the atoms in \(\alcformula\). 
    
 \begin{lemma}
\label{eqCeqF}
    Let \(M, M'\) be  models and \(\alcformula\) an \(\ALC\)-formula.
    Also let \((T, o, \f) \coloneqq \qmOf{\alcformula}{M}\) and \((T', o', \f')
    \coloneqq \qmOf{\alcformula}{M'}\). Then, 
    \(\liteqc{M}{\alcformula} \ = \  \liteqc{M'}{\alcformula}\) iff 
    \(\f = \f'\).
\end{lemma}

\begin{proof}
    First, assume that \(\liteqc{M}{\alcformula} = \liteqc{M'}{\alcformula}\).
    Then we know that for every \(\alpha \in \langlit{\alcformula}\), \(M \models
    \alpha\) iff \(M' \models \alpha\). With \cref{eqLitEqf} we can conclude
    that \(\f = \f'\).

    Now, assume that \(\f = \f'\). \cref{eqLitEqf} implies that \(lit(\f)
    = lit(\f')\). In other words, for every atomic subformula \(\alpha \in
    \langlit{\alcformula}\) we have that \(\alcformula\), \(M \models \alpha\)
    iff \(M' \models \alpha\), that is,  \(\liteqc{M}{\alcformula}
    = \liteqc{M'}{\alcformula}\).
\end{proof}

\begin{lemma}
\label{conModelsSetDiff}
Let \(M\) be a model and \(\alcformula\) an \(\ALC\)-formula. Then, the following holds:
    \(\modelsof{\alcformula} \setminus \liteqc{M}{\alcformula}
    = \modelsof{\oCon(\alcformula, M)}\).
\end{lemma}

\begin{proof}
    Let \((T, o, \f) \coloneqq \qmOf{\alcformula}{M}\) and \((T', o', \f')
    \coloneqq \qmOf{\alcformula}{M'}\).
    First,  suppose that \(M' \in \mcondiff\). We know that \(M' \models
    \alcformula\) and by \cref{IphiQMphi} we get that \(\qmOf{\alcformula}{M'}\)
    is a quasimodel for \(\alcformula\).  We also know that \(M' \not\in \ 
    \liteqc{M}{\alcformula}\). Thus, from \cref{eqCeqF}, we obtain \(\f \neq
    \f'\).  Therefore, \(\f' \in \qfilter(\alcformula, M)\). Hence, \(M' \in
    \mconstd\) and so 
     \(\mcondiff \subseteq \mconstd\).
    
    Now, let \(M' \in \mconstd\). This means that there is at least one \(\f{''} \in
    \qfilter(\alcformula, M)\) such that \(M' \models \bigwedge lit(\f'')\). But
    as consequence of the definition of formula type, this implies that \(M' \in
    \modelsof{\alcformula}\) and thus \((T', o', \f') \in \setqm{\alcformula}\). 
    We also know that \(M \not\in \ 
    \liteqc{M}{\alcformula}\), otherwise \(\f' = \f\) due to \cref{eqCeqF}.
    Therefore, \(M' \in \mcondiff\) and we can conclude that \(\mconstd
    \subseteq \mcondiff\).

    Finally, we obtain: \(\mconstd = \mcondiff\).
\end{proof}



\contractioncharacterization*

\begin{proof}


    Assume that \(\gCon(\alcformula, \modelm) \equiv
    \oCon(\alcformula, \modelm)\).
    From \cref{conModelsSetDiff} we have that
    \(\modelsof{\gCon(\alcformula, \modelm)} = \modelsof{\alcformula} \setminus
    \liteqc{\modelm}{\alcformula}\), hence success and inclusion are immediately satisfied. 
    To prove atomic retainment, assume that \(\modelm'
    \not\in \modelsof{\gCon(\alcformula, \modelm)}\) and that there is
    a  set of models
    \(\mSet'\) with \(\modelm' \in \mSet'\), \(\modelsof{\gCon(\alcformula, \modelm)}
    \subset \mSet' 
    \subseteq \modelsof{\alcformula} \setminus \liteqc{\modelm}{\alcformula}\)
    and that is finitely representable in \(\ALC\)-formula.
    \Cref{conModelsSetDiff} implies that \(\modelsof{\gCon(\alcformula,
    \modelm)} = \modelsof{\alcformula} \setminus
    \liteqc{M}{\alcformula}\). Hence, \(\modelm' \in
    \liteqc{\modelm}{\alcformula}\), a contradiction as we assumed that
    \(\mSet' \subseteq \modelsof{\alcformula} \setminus
    \liteqc{\modelm}{\alcformula}\). Therefore, no such \(\mSet'\) could exist,
    and thus, \(\gCon\) satisfies atomic retainment.

    Let \(\modelm' \equivlit \modelm\). Since \(\modelsof{\gCon(\alcformula, \modelm)} = \modelsof{\alcformula} \setminus
    \liteqc{\modelm}{\alcformula}\) and \(\liteqc{\model'}{\alcformula}
    = \liteqc{\model}{\alcformula}\), we have that: 
    \(\modelsof{\alcformula} \setminus \liteqc{\modelm}{\alcformula}
    = \modelsof{\alcformula} \setminus \liteqc{\modelm'}{\alcformula}
    = \modelsof{\gCon(\alcformula, \modelm')}\). Hence,  atomic extensionality is also satisfied.


    On the other hand, suppose that \(\gCon(\alcformula, M)\) satisfies
    the postulates stated. 
    Let \(M' \in \modelsof{\alcformula} \setminus
    \liteqc{M}{\alcformula}\) and assume that \(M' \not\in
    \modelsof{\gCon(\alcformula, M)}\). Due to atomic retainment, this means
    that there is no set \(\mSet'\) finitely representable in \(\ALC\)-formula such that 
    \(\modelsof{\gCon(\alcformula, \modelm)} \subset \mSet' \subseteq 
    \modelsof{\alcformula} \setminus \liteqc{\modelm}{\alcformula}\) 
    and \(\modelm' \in \mSet'\). But we know from \cref{conModelsSetDiff} that 
    \(\modelsof{\alcformula} \setminus \liteqc{\modelm}{\alcformula}\) is
    finitely representable in \(\ALC\)-formula and includes \(\modelm'\) by assumption, a contradiction.
    Thus, no such \(\modelm'\) could exist and 
    \(\modelsof{\alcformula} \setminus \liteqc{M}{\alcformula}
    \subseteq \modelsof{\gCon(\alcformula, M)} \).

    Now, let \(\modelm' \in \modelsof{\gCon(\alcformula, \modelm)}\).
    By inclusion \(\modelm' \in \modelsof{\alcformula}\) and by success
    \(\modelm' \neq \modelm\). We will show that
    \(\modelm' \not\in \liteqc{M}{\alcformula}\). By contradiction, suppose that
    \(\modelm' \in \ \liteqc{M}{\alcformula}\). 
    Due to atomic extensionality
    \(\modelsof{\gCon(\alcformula, M)} = \modelsof{\gCon(\alcformula, M')}\),
    but success implies that \(M' \not\in \modelsof{\gCon(\alcformula, M')}\).
    This contradicts our initial assumption that  \(M' \in
    \modelsof{\gCon(\alcformula, M)}\). Therefore \(M' \in \modelsof{\alcformula} \setminus
    \liteqc{M}{\alcformula}\) and we can conclude that 
    \(\modelsof{\gCon(\alcformula, M)} \subseteq \modelsof{\alcformula} \setminus \liteqc{M}{\alcformula}\).

    Hence, \cref{conModelsSetDiff} yields \(\oCon(\alcformula, M) \equiv
    \gCon(\alcformula, M)\).
\end{proof}


\begin{proposition}\label{prop:land_lit_f}
Let $\f$ be a formula type.  If $M \models \bigwedge lit(\f)$,  $\psi \in
\langlit{\f}$ and $M \models \psi$ then $\bigwedge lit(\f) \models \psi$.
\end{proposition}

\begin{proof}
Let $\f$ be a formula type,  $M$ a model such that $M \models \bigwedge lit(\f)$, and $\psi$ and $\ALC$-formula such that $M \models \psi$.  
The proof is by induction on the degree of $\psi$. 

\begin{itemize}
	\item[] \textbf{Base:} $degree(\psi) = 1$.  Thus, from its definition,  $\psi$ has to be an atomic formula.  
	As $\f$ is a formula type, we have that $\varphi \in \f$ iff $\neg \varphi \not \in \f$.  Let us suppose for contradiction that $\psi \not \in \f$. Thus, $\neg \psi \in \f$.  This implies that $\bigwedge \f \models \neg \psi$. Thus,  as $M \models \bigwedge lit(\f)$, we have that $M \models \neg \psi$. This contradicts the hypothesis that $M \models \psi$.  Thus, we conclude that $\psi \in \f$.  Therefore, $\bigwedge lit(\f) \models \psi$.
	
	\textbf{Induction Hypothesis:} For every formula $\varphi$,  if $degree(\varphi) < degree(\psi)$ and   $M \models \varphi$ then $\bigwedge lif(\f) \models \varphi$.

 \textbf{Induction Step:} Let $degree(\psi) > 1$.  By construction, $\psi$ is of the form (1) $\varphi \land \varphi'$ or (2) $\neg \varphi$, for some $\ALC$-formulae $\varphi$ and $\varphi'$: 
 \begin{enumerate}
 	\item[] (1) $\psi = \varphi \land \varphi'$.   From definition, $degree(\varphi \land \varphi') = degree(varphi) + degree(\varphi')$. Recall from definition of $degree$ that $degree(\beta) > 1$, for every formula $\beta$. Therefore, 
 	$degree(\varphi) < degree(\varphi \land \psi')$ and $degree(\varphi') < degree(\varphi \land \varphi')$. This means that $degree(\varphi) < degree(\psi)$ and  $degree(\varphi') < degree(\psi)$.  From hypothesis,  $M \models \psi = \varphi\land \varphi'$.   Thus,  $M \models \varphi$ and $M \models \psi$. This implies from IH that 
$$ \bigwedge lit(\f) \models \varphi \mbox{ and } \bigwedge lit(\f) \models \varphi'  $$
Therefore,  $ \bigwedge lit(\f) \models \varphi \land \varphi' = \psi$.  

	\item[] (2) $\psi = \neg \varphi$. We have two cases,  either (i) $\varphi$ is an atomic formula or (ii) $\varphi = (\beta \land \beta')$.   For the first case,  analogous to the base case,  we get that $\bigwedge lit(\f) \models \psi$.  So we focus only on the second case.  From the definition of $degree$, we get that $degree(\neg \beta)< degree(\psi)$ and $degree(\neg beta') < degree(\psi)$.  As $M \models \psi = \neg(\beta \land \beta')$, we get that either (a) $M \models \neg \beta$ or (b) $M \models \neg \beta'$. 
		\begin{itemize}
			\item[] \textbf{(a)} $M \models \neg \beta$. From above,  $degree(\neg \beta) < degree(\psi)$. Thus,  from IH,  we get that $\bigwedge lit(\f) \models \neg \beta$.  Thus,  $\bigwedge lit(\f) \models \neg(\beta \land \beta') = \psi$.
			\item[] \textbf{(b)}: $M \models \neg \beta'$.  Analogous to case
                (a). 
		\end{itemize}	
	
 \end{enumerate}

\end{itemize}

\end{proof}

\begin{lemma}\label{lem:lit_equals_neg}
If $M \models \varphi$,  $\f \in \ftypes{\varphi}$ and $M \models \bigwedge
lit(\f)$ then $\modelsof{\bigwedge lit(\f)} = \liteqc{M}{\alcformula}$. 
\end{lemma}

\begin{proof}
We need to show that $M' \in \modelsof{\bigwedge lit(\f)}$ iff $M' \in \liteqc{M}{\alcformula}$.  
\begin{itemize} 
	\item[]``$\Rightarrow$''.  $M' \in \modelsof{\bigwedge lit(\f)}$.  To show
        that $M' \in \liteqc{M}{\alcformula}$, it suffices to show that $M' \equivlit M$.  Let $\psi \in \langlit{\varphi}$, we need to show that $M \models \psi$ iff $M \models \psi$.
	\begin{itemize}
		\item[] (a) ``$\Rightarrow$''  $M \models \psi$.  From Proposition~\ref{prop:land_lit_f}, we have that 
		$\bigwedge lit(\f) \models \psi$.  
     This jointly with $M' \models \bigwedge lit(\f)$ implies that  $M' \models \psi$.  
		\item[] (b) ``$\Leftarrow$'' $M' \models \psi$.  
		Analogous to item (a).  
	\end{itemize}
	\item[] ``$\Leftarrow$'' $M' \in \liteqc{M}{\alcformula}$.  Note that $\bigwedge lit(\f) \in \langlit{\varphi}$.  Thus as $M' \equivlit M$,  and $M \models \bigwedge lit(\f)$, we get that $M' \models \bigwedge lit(\f)$.  
\end{itemize}
\end{proof}

\expMmodelUnion*
\begin{proof}
We have two cases: either (i) $M \models \varphi$ or (ii)$M \not \models \varphi$.  
\begin{enumerate}
	\item[] (i) $M \models \varphi$.  Then,  from definition of $\Ex$, we have that $\Ex(\varphi, M) = \varphi$ which implies that
	 $\modelsof{\Ex(\varphi,M)} = \modelsof{\varphi}$.  As $M \models \varphi$,
     we get that $\liteqc{M}{\alcformula} \subseteq \modelsof{\varphi}$.  Therefore,
     $\modelsof{\varphi} \cup \liteqc{M}{\alcformula} = \modelsof{\varphi}$. This
     implies that $$\modelsof{\Ex(\varphi,M)} = \modelsof{\varphi} \cup
     \liteqc{M}{\alcformula}. $$
	\item[] (ii)$M \not \models \varphi$.  Thus, from definition of $\Ex$, we get that 
	$$\Ex(\varphi,M) = \varphi \lor \bigwedge lit(\f),  \mbox{ where } qm(\neg \varphi, M) = (T,o,\f).$$  
	This implies that 
	$ \modelsof{\Ex(\varphi,M)} =   \modelsof{\varphi \lor \bigwedge lit(\f)}$.  
	Note that $$ \modelsof{\varphi \lor \bigwedge lit(\f)} = \modelsof{\varphi}
    \cup \modelsof{\bigwedge lit(\f)}.  $$
	
	As $qm(\neg \varphi, M) = (T,o,\f)$,  it follows from the definition of $qm$ that $\f \in \ftypes{\neg \varphi}$ and 
	$M \models \bigwedge lit(\f)$.   In summary,  $M \models \neg \varphi$, $\f \in \ftypes{\neg \varphi}$ and $M \models \bigwedge lit(\f)$.  Thus,  from Lemma~\ref{lem:lit_equals_neg},  we have that 
	$$\modelsof{\bigwedge lit(\f)} = \liteqc{M}{\alcformula}. $$
	Therefore,    
	$$ \modelsof{\Ex(\varphi,M)} = \modelsof{\varphi} { \cup } \liteqc{M}{\alcformula}.  $$
	
\end{enumerate}
 \end{proof}

\expModelConverse*
\begin{proof}
We consider each case separately.  
\begin{enumerate}
	\item[] (i)   $M \models \varphi$.  
		Thus,  $\liteqc{M}{\alcformula}  {\subseteq } \ \modelsof{\varphi}$ which
        implies that $$\modelsof{\varphi} \cup \liteqc{M}{\alcformula}
        = \modelsof{\varphi}.$$ Therefore,  $\Ex^*(\varphi,M) \equiv \varphi$.  
	\item[] (ii) $M \not \models\varphi$.   Let $qm(\neg \varphi, M)
        = (T,o,\f)$.  Note that  $\f \in \ftypes{\varphi}$, and from definition
        of $qm$ that $M \models \bigwedge lit(\f)$. Thus,  it follows from
        Lemma~\ref{lem:lit_equals_neg} that $\modelsof{\bigwedge lit(\f)} = \liteqc{M}{\alcformula}$.  
	Thus,  $\modelsof{\varphi \lor \bigwedge qm(\neg \varphi,M)} = \modelsof{\varphi} \cup \liteqc{M}{\alcformula}$. This means that $\Ex^*(\varphi,M) \equiv \varphi \lor \bigwedge qm(\neg \varphi,M)$.  
\end{enumerate}

\end{proof}

\expansioncharacterization*

\begin{proof}
    First, assume that \(\gExp(\alcformula, \modelm) \equiv
    \oExp(\alcformula, \modelm)\).
    From \cref{exp_model_union} we have that
    \(\modelsof{\gExp(\alcformula, \modelm)} = \modelsof{\alcformula} \cup
    \liteqc{\modelm}{\alcformula}\), hence success and persistence are immediately satisfied. 
    To prove atomic temperance, assume that \(\modelm'
    \in \modelsof{\gExp(\alcformula, \modelm)}\) and that there is a set of models
    \(\mSet'\) with \(\modelm'\) that is finitely representable in
    \(\ALC\)-formula and such that
    \(\modelsof{\alcformula} \cup \liteqc{\modelm}{\alcformula}
    \subseteq \mSet' 
    \subset \modelsof{\gExp(\alcformula, \modelm)}\).
    \Cref{exp_model_union} implies that \(\modelsof{\gExp(\alcformula,
    \modelm)} = \modelsof{\alcformula} \cup
    \liteqc{M}{\alcformula}\). Hence, \(\modelm' \not\in
    \liteqc{\modelm}{\alcformula}\), a contradiction as we assumed that
    \(\mSet' \supseteq \modelsof{\alcformula} \cup
    \liteqc{\modelm}{\alcformula}\). Therefore, no such \(\mSet'\) could exist,
    and thus, \(\gExp\) satisfies atomic temperance.

    Let \(\modelm' \equivlit \modelm\). Since \(\modelsof{\gExp(\alcformula,
    \modelm)} = \modelsof{\alcformula} \cup
    \liteqc{\modelm}{\alcformula}\) and \(\liteqc{\model'}{\alcformula}
    = \liteqc{\model}{\alcformula}\), we have that: 
    \(\modelsof{\alcformula} \cup \liteqc{\modelm}{\alcformula}
    = \modelsof{\alcformula} \cup \liteqc{\modelm'}{\alcformula}
    = \modelsof{\gExp(\alcformula, \modelm')}\). Hence, atomic extensionality is also satisfied.

    On the other hand, suppose that \(\gExp(\alcformula, M)\) satisfies
    the postulates stated. 
    Let \(\modelm' \in \modelsof{\alcformula} \cup
    \liteqc{\modelm}{\alcformula}\). If \(\modelm' \in \modelsof{\alcformula}\)
    then success ensures that 
    \(\modelm' \in \modelsof{\gExp(\alcformula, \modelm)}\). Otherwise, we have
    \(\modelm' \equivlit \modelm\), and as consequence of success and atomic
    extensionality we also obtain \(\modelm' \in \modelsof{\gExp(\alcformula,
    \modelm)}\). Therefore, \(\modelsof{\alcformula} \cup
    \liteqc{\modelm}{\alcformula} \subseteq \modelsof{\gExp(\alcformula,
    \modelm)}\).

    Now, let \(\modelm' \in \modelsof{\gExp(\alcformula, \modelm)}\) and assume
    that \(\modelm' \not\in \modelsof{\alcformula} \cup
    \liteqc{\modelm}{\alcformula}\). Success, persistence and atomic
    extensionality imply that \(\modelsof{\gExp(\alcformula, \modelm)}\). Atomic
    temperance states that there is no set of models \(\mSet'\) that is finitely
    representable in \(\ALC\)-formula with 
    \(\modelsof{\alcformula} \cup \liteqc{\modelm}{\alcformula} \subseteq \mSet' \subset
    \modelsof{\gExp(\alcformula, \modelm)} \cup \{\modelm\}\). But we know from
    \cref{exp_model_union} that \(\modelsof{\alcformula} \cup
    \liteqc{\modelm}{\alcformula}\) is finitely representable in
    \(\ALC\)-formula and does not include
    \(\modelm'\) by assumption, a contradiction. Thus, no such \(\modelm'\)
    could exist and \(\modelsof{\gExp(\alcformula,
    \modelm)} \subseteq \modelsof{\alcformula} \cup \liteqc{\modelm}{\alcformula}\).

    Hence, \cref{exp_model_union} yields \(\gExp(\alcformula, \modelm) \equiv
    \oExp(\alcformula, \modelm)\).
\end{proof}